\documentclass[journal]{IEEEtran}
\usepackage{amsmath,amssymb,amsfonts}
\usepackage{algorithmic}
\usepackage{booktabs}
\usepackage{array}
\usepackage[subrefformat=parens,labelformat=parens,caption=false,font=footnotesize,labelfont=rm,textfont=rm]{subfig}
\usepackage{textcomp}
\usepackage{stfloats}
\usepackage{url}
\usepackage{bm}
\usepackage{verbatim}
\usepackage{multirow}
\usepackage{graphicx}
\usepackage{cite}
\usepackage{balance}
\usepackage{empheq}
\usepackage{amsthm}

\usepackage{float}

\theoremstyle{remark}

\theoremstyle{plain}
\newtheorem{theorem}{Theorem}

\usepackage{tikz}
\usepackage[subnum]{cases}

\usepackage[linesnumbered,ruled,lined]{algorithm2e}

\ifCLASSINFOpdf

\else

\fi

\begin{document}
\title{An Efficient Max-Min Fair Resource Optimization Algorithm for Rate-Splitting Multiple Access}
\author{Facheng Luo, Yijie Mao,~\IEEEmembership{Member,~IEEE}

\thanks{
\par
This work has been supported in part by the National Nature Science Foundation of China under Grant 62201347; and in part by Shanghai Sailing Program under Grant 22YF1428400. 
\par 
F. Luo and Y. Mao are with the School of Information Science and Technology, ShanghaiTech University, Shanghai 201210, China (email: \{luofch2022, maoyj\}@shanghaitech.edu.cn).}
}

\maketitle

\begin{abstract}
    The max-min fairness (MMF) problem in rate-splitting multiple access (RSMA) is known to be challenging due to its non-convex and non-smooth nature,
    as well as the coupled beamforming and common rate variables.
    Conventional algorithms to address this problem often incur high computational complexity or degraded MMF rate performance.
    To address these challenges, in this work,
    we propose a novel optimization algorithm named extragradient-fractional programming (EG-FP) to address the MMF problem of downlink RSMA.
    The proposed algorithm first leverages FP to transform the original problem into a block-wise convex problem.
    For the subproblem of precoding block,
    we show that its Lagrangian dual is equivalent to a variational inequality problem,
    which is then solved using an extragradient-based algorithm.
    Additionally, we discover the optimal beamforming structure of the problem and based on which,
    we introduce a low-dimensional EG-FP algorithm with computational complexity independent of the number of transmit antennas. This feature is especially beneficial in scenarios with a large number of transmit antennas.
    The proposed algorithms are then extended to handle imperfect channel state information at the transmitter (CSIT).
    Numerical results demonstrate that the  MMF rate achieved by our proposed algorithms closely matches that of the conventional successive convex approximation (SCA) algorithm and significantly outperforms other baseline schemes.
    Remarkably, the average CPU time of the proposed algorithms is less than 10\% of the runtime required by the SCA algorithm, showing the efficiency and scalability of the proposed algorithms.
\end{abstract}

    
\IEEEpeerreviewmaketitle

\section{Introduction}
\label{introduction}
\IEEEPARstart{T}{he} next generation mobile communication system has attracted immense interest from both academic and industrial spheres.
As the motivation for the Internet of everything,
the next generation mobile communication system is expected to revolutionize services by offering higher throughput,
ultra-reliable connectivity, and diverse quality of service (QoS) while supporting massive connectivity \cite{mao2022rate}.
With the expansive array of devices, deployments, and applications in 6G,
the challenge lies in catering to users with varying capabilities and diverse qualities of channel state information at the transmitter (CSIT) simultaneously.
Consequently, the effective utilization of wireless resources and the management of interference are imperative for 6G, prompting a comprehensive revolution of physical layer resource allocation algorithms and multiple access (MA) techniques for wireless communication systems.
\par
Rate-splitting multiple access (RSMA), an innovative interference management and MA strategy, is recognized as a powerful technique for future wireless communication systems by employing linearly precoded rate-splitting (RS) \cite{commag2016clerckx, mao2022rate, clerckx2023RSMAtutorial}.
This method contrasts with existing MAs such as space-division multiple access (SDMA) which treats interference as noise entirely, and non-orthogonal multiple access (NOMA) which fully decodes interference.
RSMA bridges between these extremes by offering an adaptive interference management strategy.
It intelligently splits users' messages into distinct common and private parts.
Common parts are collaboratively encoded and decoded among multiple users,
while private parts are individually decoded by their respective users.
Built upon this design principle, 
RSMA effectively enables partial decoding of interference while treating the remaining interference as noise,
thereby surpassing the performance achieved by SDMA and NOMA \cite{mao2018rate}.

\par 
The joint beamforming vectors (including power allocation) and common rate optimization for RSMA have been extensively studied in both perfect and imperfect CSIT across various design objectives.
These objectives include, but are not limited to, maximizing the weighted sum-rate\cite{xu2021},
achieving the max-min rate\cite{luo2023practical},
enhancing energy efficiency\cite{zhangruichen2023},
and minimizing overall power consumption\cite{camana2021}.
In addition to these downlink RSMA studies, numerous uplink RSMA studies also exist.
For example, \cite{yang2020sum} transforms the non-convex sum-rate maximization problem in uplink RSMA into a solvable rate-based problem with a closed-form solution, derives the optimal power and decoding order using exhaustive search, and proposes a low-complexity user pairing scheme to enhance efficiency.
\cite{tegos2022performance} derives closed-form expressions for the outage probability of user messages. These expressions are then leveraged to calculate the throughput  in a random access network based on  RSMA.
\cite{chrysologou2024coexistence} derives novel closed-form expressions for outage probability and ergodic rate, as well as system-wide outage throughput and ergodic sum rate.
One of the most common design objectives in a multiuser scenario is to maximize the system sum-rate or weighted sum-rate.
However, as highlighted in \cite{Joudeh2016sumrate, fang2024rate}, in the context of RSMA sum-rate optimization,
it is possible for certain users to experience enforced zero rates in order to achieve the maximum sum-rate.
To ensure fairness among users, max-min fairness (MMF) has become another popular resource optimization problem in RSMA.
This strategy aims to balance user experiences by maximizing the worst-case rate among users, thereby ensuring a fair and equitable distribution of rates among users.
\par 
Due to the inherent non-convex and non-smooth nature of MMF resource allocation problem,
a number of algorithms have been proposed to tackle these challenges.
Notable among these are the weighted minimum mean-square error (WMMSE) algorithm \cite{yalcin2021, joudeh2017rate},
successive convex approximation (SCA) algorithm \cite{xuyunnuo2023max, mao2020max},
and the generalized power iteration (GPI) algorithm \cite{park2022rate, korean}.
The WMMSE algorithm addresses non-convexity introduced by rate expressions through the rate-WMMSE relationship and tackles non-smoothness introduced by the max-min rate using auxiliary variables.
It successively solves the MMF problem of RSMA by iteratively using a solver in optimization toolboxes to solve a convex subproblem. 
Specifically, considering a RSMA-assisted cooperative multigroup multicasting network,
\cite{yalcin2021} developed an alternating optimization algorithm based on WMMSE to optimize precoders that maximize the minimum achievable rate across multicast groups under high user demand and interference conditions.
\cite{joudeh2017rate} further solves the MMF problem of RSMA in overloaded multigroup multicast beamforming systems based on WMMSE.
Similarly, the SCA algorithm addresses the MMF problem of RSMA by constructing a series of convex surrogate functions, which are solved iteratively using optimization toolboxes. 
Existing works, such as \cite{xuyunnuo2023max},
proposes a 1D-SCA algorithm to optimize time allocation and precoders for maximizing the minimum user rate in downlink multi-antenna RSMA systems with user relaying,
under finite blocklength constraints.
\cite{mao2020max} extends this by introducing a two-stage approach:
first selecting relaying users based on channel gains,
then jointly optimizing precoders, common rate,
and time slots using SCA to ensure fairness in a $K$-user MISO broadcast channel under a power constraint.
Both WMMSE and SCA require a relatively high computational complexity, owing to their reliance on optimization toolboxes to solve a sequence of problems.

Recently, a novel GPI algorithm is proposed as an optimization toolbox-free algorithm to address the MMF problem of RSMA.
GPI operates via a two-step iterative approach. It initially adopts the LogSumExp (LSE) technique to optimize the beamformers for a given common rate allocation.
Subsequently, it determines the common rate allocation under given beamformers by employing a water-filling-like method.
The two steps are iterated until convergence. 
However, unlike WMMSE and SCA, the convergence of GPI to local optima remains uncertain due to its use of fixed common rate allocation from the initial step. This limits the flexibility for subsequent updates of beamformers and common rates.
Another recent work \cite{dizdar2024maxmin} introduces a low-complexity algorithm for overloaded networks.
The algorithm fixes the beamformers by zero-forcing (ZF) and maximum ratio transmission (MRT),
and then optimizes the power allocation among streams.
Although the algorithm reduces computational complexity,
it shrinks the feasible solution space and does not guarantee convergence to a locally optimal solution of the original RSMA MMF problem,
where beamforming vectors and rate allocations are jointly optimized.

\par 
In addition to the aforementioned works, our recent research  \cite{luo2023practical} introduces a groundbreaking closed-form solution addressing the MMF problem of RSMA.
By fixing the directions of beamformers based on zero-forcing, and using water-filling for power allocation among the private streams,  the proposed algorithm identifies the optimal power allocation between the common and private streams. This is achieved by evaluating six potential candidates and selecting the one that maximizes the MMF rate without the need for an iterative process.
This approach has been demonstrated to achieve 90\% of the MMF rate of the SCA algorithm, requiring only 0.1 milliseconds of computational time.
However, it is restricted to the two-user case with perfect CSIT.

\par
Indeed,  the non-convex and non-smooth nature of the RSMA MMF problem prevents the derivation of closed-form  optimal beamforming solutions. 
Existing approaches  mostly rely on interior-point method, typically implemented through optimization toolboxes like CVX,
 to achieve locally optimal solutions.
However, these methods incur substantial computational complexity,
as highlighted in studies such as \cite{xuyunnuo2023max, mao2020max, yalcin2021, joudeh2017rate}.
Meanwhile, alternative methods that avoid the interior-point method often result in degraded MMF rate performance,
as seen in \cite{park2022rate, korean, dizdar2024maxmin}.
Therefore, our work aims to bridge this research gap by introducing an efficient algorithm that converges to a locally optimal solution for the MMF problem of RSMA while significantly reducing the computational complexity for downlink RSMA.
This is achieved by leveraging the FP approach and the extragradient algorithm \cite{marcotte1991application, zhang2023ultra}.
The primary contributions of this paper are summarized as follows:
\begin{itemize}
    \item
    We propose a novel optimization algorithm named the extragradient-based FP (EG-FP) algorithm to solve the MMF problem of RSMA.
    The original problem is first transformed into a sequence of convex subproblems based on FP.
    By exploring the Lagrangian dual problem of each convex subproblem, we equivalently transform the dual problem into a variational inequality problem. The main novelty of the proposed EG-FP algorithm lies in effectively utilizing the variational inequality property to address each subproblem via the extragradient algorithm.  

    \item 
    We discover the optimal beamforming structure of the MMF problem by leveraging the Karush-Kuhn-Tucker (KKT) conditions.
    This exploration has revealed that:  the optimal beamforming vectors of RSMA are linear combinations of channel vectors. Building upon this valuable insight, we propose a low-dimensional variant of our proposed EG-FP algorithm. This algorithm achieves low computational complexity, which is independent of the number of transmit antennas. It therefore shows significant computational advantages in scenarios involving a massive number of transmit antennas.

    \item 
    In addition to addressing the MMF problem of RSMA for perfect CSIT,
    our contributions also extend to addressing the imperfect CSIT scenario.
    The formulations and methodologies developed in our approach apply to the imperfect CSIT scenario,
    thereby broadening the applicability and robustness of our proposed algorithm.

    \item 
    To illustrate the validity and efficiency of our novel algorithms, comprehensive numerical simulations have been carefully conducted in both small-scale and large-scale networks.
    These results not only show the effectiveness but also the computational efficiency of our proposed approach, underlining the practical viability of the proposed algorithms in real-world applications.
    
\end{itemize}
To the best of our knowledge,
this is the first work that discovers locally optimal solutions for the joint optimization of beamforming vectors and common rate allocation of RSMA without relying on any optimization toolbox.
While the algorithm proposed in \cite{fang2024rate} for the weighted sum-rate problem of RSMA appears to encompass common rate allocation,
the inherent nature suggests that common rate allocation can be eliminated.
Consequently, the algorithm cannot apply to the MMF problem of RSMA.
The algorithms proposed in this work not only address the MMF problem but also have the potential to handle other RSMA problems where common rate allocation cannot be omitted.

\par 
The paper is organized as follows. 
Section \ref{system model} details the system model and formulates the MMF problem of RSMA.
Section \ref{algorithm} presents the optimization framework designed to solve the formulated problem.
Section \ref{imperfect CSIT} generalizes the proposed algorithm to the imperfect CSIT scenario. 
Following this, Section \ref{numerical results} presents and analyzes the outcomes of the numerical simulations.
Finally, Section \ref{conclusion} concludes the paper.

\textit{Notation:} We denote matrices using upper-case bold letters,
while lower-case bold letters represent vectors.
The conjugate transpose of a vector or matrix is denoted by $(\cdot)^H$,
while the transpose is indicated by $(\cdot)^T$.
The $l_2$-norm for a vector is represented by $\|\cdot\|$, and $|\cdot|$ denotes the modulus of a complex number.
The real and imaginary parts of a complex number are respectively denoted by $\Re\{\cdot\}$ and $\Im\{\cdot\}$.
The conjugate of a complex number or vector is denoted by $(\cdot)^*$.
The mathematical expectation operator is represented by $\mathbb{E}\{\cdot\}$.
Furthermore, $\operatorname{tr}(\cdot)$ denotes the trace of a matrix.

\section{System Model and Problem Formulation}
\label{system model}
\begin{figure}[htbp]
    \centering
    \includegraphics[width=\linewidth]{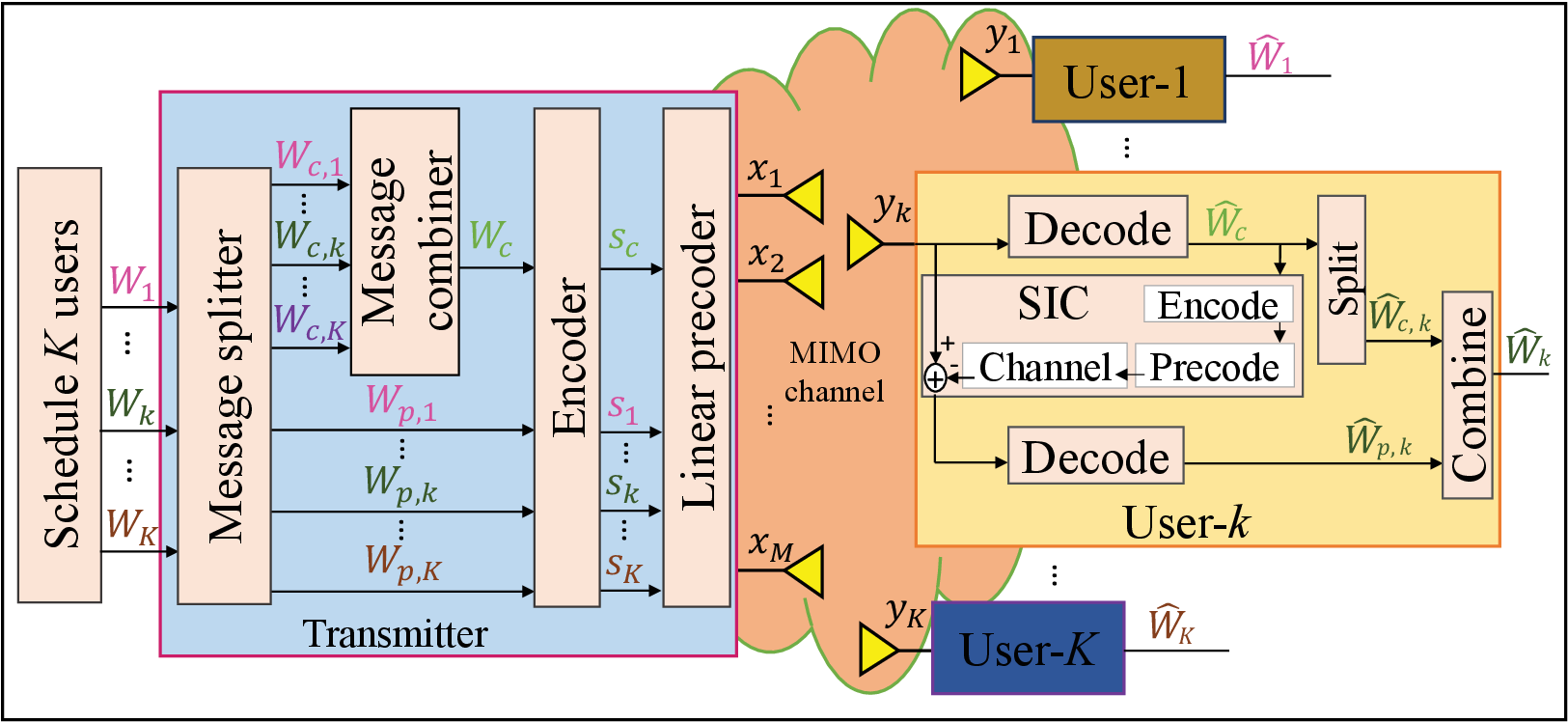}
    \caption{{System model of rate-splitting multiple access}}
\end{figure}
We consider the downlink transmission in a single-cell multi-user multiple-input single-output (MU-MISO) cellular network, comprising a base station (BS) and multiple users.
The BS, equipped with $N_t$ transmit antennas, simultaneously serves a set of $K$ users, each having only one single receive antenna, labeled as $\mathcal{K}=\{1, 2, \ldots, K\}$.
1-layer RS, as proposed in \cite{commag2016clerckx}, is supported during the transmission.
\par 
At the BS,  message $W_k$ of user-$k$ gets split into two parts: a common part $W_{c,k}$ and a private part $W_{p,k}$.
The common parts from all users are collected together and encoded into a unified common stream referred to as $s_c$.
Meanwhile, the private parts are separately encoded into private streams $s_1, \cdots, s_K$.
We assume the streams ${s_c,s_1,\cdots, s_K}$ are independent and follow the distribution $\mathcal{CN}(0, 1)$, i.e., $\mathbb{E}\{s_ks_k^H\} = 1$
and $ \mathbb{E}\{s_ks_j^H\} = 0, \forall k\neq j$, and $k,j\in\{c, 1, 2, \ldots, K\}$.
Each stream $s_k$ is linearly precoded by a beamforming vector $\mathbf{p}_k\in\mathbb{C}^{N_t\times 1}$,
collectively forming the beamforming matrix for all streams, denoted as $\mathbf{P}=[\mathbf{p}_c,\mathbf{p}_1,\cdots,\mathbf{p}_K] \in \mathbb{C}^{N_t\times (K+1)}$. The precoded streams are then superposed and transmitted to the users. The resulting transmit signal sent by the BS is:
\begin{equation}
    \mathbf{x}= \mathbf{p}_c s_c + \sum\nolimits_{j=1}^{K}{\mathbf p}_j s_j.
\end{equation}
The transmit power constraint is given by $ \mathrm{tr}(\mathbf{P}^H\mathbf{P})\leq P_t $,
where $ P_t $ refers to the maximum power the BS can transmit.
\par
The channel vector between the BS and user $k$ is represented by $\mathbf{h}_k\in\mathbb{C}^{N_t\times 1}$.
We assume that the channel responses remain constant throughout one transmission block,
and the BS has access to perfect channel state information (CSI).
The scenario of imperfect CSIT is explored in Section \ref{imperfect CSIT}.
The  signal received by user $k\in\mathcal{K}$ can be expressed as:
\begin{equation}
y_k=\mathbf{h}_k^H\mathbf{p}_c s_c + \mathbf h_k^H\sum\nolimits_{j=1}^{K}{\mathbf p}_j s_j + n_k.
\end{equation}
Here, $n_k\sim \mathcal{CN}(0,\sigma_k^2)$ represents the additive white Gaussian noise (AWGN) at user $k$, where $\sigma_k^2$ denotes the noise power.
\par
Each user first focuses on decoding the common stream $s_c$,
treating the private streams as interference.
Then, employing successive interference cancellation (SIC), each user removes the common stream $s_c$ from their received signals.
Afterward, each user decodes the intended private stream, treating interference from other streams as noise.
The signal-to-interference-plus-noise ratio (SINR) for decoding the common stream $s_c$ at user-$k$ is calculated by:
\begin{equation}
\label{eq:SINRck}
    \gamma_{c, k} = \frac{|\mathbf h_k^H\mathbf p_c|^2}
    {\sum\limits_{j=1}^K|\mathbf h_k^H\mathbf p_j|^2 + \sigma_k^2}.
\end{equation}
The SINR for decoding the private stream $s_k$ at user-$k$ is:
\begin{equation}
\label{eq:SINRpk}
    \gamma_{p, k} = \frac{|\mathbf h_k^H\mathbf p_k|^2}
    {\sum\limits_{j\neq k}^K|\mathbf h_k^H\mathbf p_j|^2 + \sigma_k^2}.
\end{equation}
The instantaneous rates for decoding the common and private streams at user-$k$ are expressed as:
\footnote{For ease of notation, we use the natural logarithm $\log(\cdot)$ instead of the binary logarithm $\log_2(\cdot)$ in this paper. The analysis is readily extended to the binary logarithm case.}
\begin{subequations}
    \begin{align}
        R_{c, k} = \log \left(1+\gamma_{c, k}\right), \label{3}\\
        R_{p, k} = \log \left(1 + \gamma_{p, k}\right). \label{4}
\end{align}    
\end{subequations}
To ensure successful decoding of the common stream $s_c$ by all users,
its rate must not exceed:
\begin{equation}
    R_c = \min\limits_{k\in \mathcal{K}}\left\{R_{c, k}\right\}.
\end{equation}
Let $c_k$ denote the portion of $R_c$ assigned to user-$k$,
following $\sum_{k = 1}^K c_k \leq R_c$ and $c_k \geq 0$ for all $k\in\mathcal{K}$.
The total achievable rate for user-$k$ is then given by $R_k + c_k$. Therefore, the worst-case rate among all users is denoted as $R_{\text{MMF}} = \min_{k\in\mathcal{K}}\{R_k + c_k\}$.
\par 
To ensure fairness among users, our paper focuses on solving the following MMF optimization problem:
    \begin{subequations}\label{P1}
        \begin{align}
            \max_{\mathbf P, \mathbf c} &\min\limits_{k\in\mathcal{K}}\{R_{p, k} + c_k\} \\
            \mathrm{s.t.}\ &\mathbf{1}^T\mathbf{c} \leq R_c, \label{P1 b}\\
            &\mathbf{c} \succeq \mathbf{0}, \label{P1 c}\\
            &\mathrm{tr}\left(\mathbf{P}^H\mathbf{P}\right) \leq P_t \label{P1 d},
        \end{align}
    \end{subequations}
where $\mathbf{c}=[c_1,c_2,\ldots,c_K]$. (\ref{P1 b}) and (\ref{P1 c}) are the constraints for the portion of $R_c$ assigned to each user, and (\ref{P1 d}) is the transmit power constraint.

Problem (\ref{P1}) is particularly challenging due to its non-convex and non-smooth nature,
which makes it different from conventional non-convex minimization problems. Specifically, the non-convexity of problem (\ref{P1}) arises from the SINR expressions $\gamma_{c, k}$ and $\gamma_{p, k}$, while the non-smoothness arises from two sources:
the minimum operation in the common rate expression,
$R_c = \min_{k\in \mathcal{K}}\left\{R_{c, k}\right\}$,
and the worst-case rate among users in the objective function,
$\min_{k\in \mathcal{K}}\{R_k + c_k\}$.
This non-smoothness further complicates the optimization process by introducing a lack of differentiability,
which makes standard techniques relying on smoothness—such as those based on continuous gradients or well-behaved objective functions—ineffective.
Various algorithms have emerged to solve problem (\ref{P1}),
such as WMMSE\cite{yalcin2021, joudeh2017rate}, SCA\cite{xuyunnuo2023max, mao2020max}, 
and GPI\cite{korean}.
These algorithms are briefly summarized below:
\begin{itemize}
    \item \textbf{WMMSE:}
    The WMMSE algorithm converts the non-convex non-smooth MMF problem into block-wise convex problems using the rate-WMMSE relation and auxiliary variables.
    Each block is solved by the solvers in optimization toolboxes iteratively until reaching convergence \cite{yalcin2021, joudeh2017rate}.
    
    \item \textbf{SCA:}
    The SCA algorithm approximates both the non-convex objective function and constraints by a series of convex functions and solves them iteratively until convergence.
    Similar to WMMSE, an optimization toolbox is needed to solve each subproblem \cite{xuyunnuo2023max, mao2020max}.
    
    \item \textbf{GPI:}
    The GPI algorithm approximates the non-smooth minimum function by employing the LSE technique to transform it into a smooth function. Despite this approximation, the resulting problem remains non-convex. This non-convexity is then addressed by the generalized power iteration \cite{korean}.
    
\end{itemize}
Although WMMSE and SCA can achieve a local optimum for problem (\ref{P1}), their computational complexities remain high, primarily due to the inability to derive closed-form solutions for the specific RSMA MMF problem considered in this work. As a result, these methods typically rely on iterative techniques such as the interior-point method, which is commonly implemented using optimization toolboxes like CVX.
While GPI does not depend on any toolbox, it cannot ensure reaching a local optimum due to its use of fixed common rate allocation from the initial step.
The exhaustive search for Lagrangian multipliers also adds to the computational complexity of GPI. 
Due to these limitations, this paper introduces a novel approach using FP and the extragradient algorithm to address the non-convex non-smooth MMF problem (\ref{P1}).

\section{Proposed Optimization Framework}
\label{algorithm}
In this section, we introduce a novel algorithm to tackle problem (\ref{P1}).
To effectively address the complexities inherent in this problem,
we first leverage the concept of FP,
an optimization technique allowing us to transform the original problem into a sequence of convex subproblems.
However, instead of directly resolving these transformed converted convex subproblems, we resort to their Lagrangian dual problems. This helps us unveil a prominent insight: \textit{the Lagrangian dual problem is equivalent to a variational inequality problem}.
Inspired by this insight, we propose an extragradient-based algorithm to effectively address the Lagrangian dual problem. 
Additionally, we provide a refined and low-dimensional optimization framework specifically designed for massive MIMO applications.
This tailored approach serves as a practical solution to enhance computational efficiency in such settings.

\subsection{Fractional Programming Approach}
\label{sec:FP}
Due to its non-convex and non-smooth nature,  solving problem (\ref{P1}) directly is highly challenging.
To address this, in this subsection, we first employ the FP technique to convert problem (\ref{P1}) into a series of convex subproblems.
Then, we focus on solving each convex subproblem.
\par
Employing the Lagrangian dual method introduced in \cite{FPpartone, FPparttwo} to extract terms from the logarithm expression,
we establish lower bounds for the achievable rates of user-$k$.
To facilitate the derivation, we introduce auxiliary variables $\vartheta_{c,k}$ and $\vartheta_{p,k}$ corresponding to the SINR values $\gamma_{c,k}$ and $\gamma_{p,k}$, respectively.
This allows us to define the following lower-bound functions:
\begin{subequations}
   \small
    \begin{align}
        f_{c,k}(\mathbf P,\vartheta_{c,k})\triangleq \log(1+\vartheta_{c,k})-\vartheta_{c,k}+\frac{(1+\vartheta_{c,k})\gamma_{c,k}}{1+\gamma_{c,k}}, \label{eq:fck}\\
        f_{p,k}(\mathbf P,\vartheta_{p,k})\triangleq \log(1+\vartheta_{p,k})-\vartheta_{p,k}+\frac{(1+\vartheta_{p,k})\gamma_{p,k}}{1+\gamma_{p,k}}. \label{eq:fpk}
    \end{align}
\end{subequations}
These functions, $f_{c, k}$ and $f_{p, k}$, act as lower bounds for $R_{c, k}$ and $R_{p, k}$, respectively, i.e.,
\begin{subequations}\label{f < r}
\small
    \begin{align}
        f_{c, k} \leq R_{c, k}, \\
        f_{p, k} \leq R_{p, k},
    \end{align}
\end{subequations}
and the equalities hold when $\frac{\partial f_{c, k}}{\partial \vartheta_{c, k}} = 0$ and $\frac{\partial f_{p, k}}{\partial \vartheta_{p, k}} = 0$, leading to the following optimal values: 
\begin{subequations}\label{optimal theta}
\small
    \begin{align}
        \vartheta_{c, k}^\circ = \gamma_{c, k}, \\
        \vartheta_{p, k}^\circ = \gamma_{p, k}.
    \end{align}
\end{subequations}
However, handling the fractional components in (\ref{eq:fck}) and (\ref{eq:fpk}) is still challenging.
By further leveraging the quadratic transform detailed in \cite{FPparttwo}, we derive alternative lower bounds for $f_{c, k}$ and $f_{p, k}$.
We define new functions $g_{c,k}(\mathbf P,\vartheta_{c,k},\varphi_{c,k})$ and $g_{p,k}(\mathbf P,\vartheta_{p,k},\varphi_{p,k})$ as follows:
\begin{subequations}
\small
    \begin{align}
        &\begin{aligned}
            g_{c,k}&(\mathbf P,\vartheta_{c,k},\varphi_{c,k})\\
            &\triangleq \log(1+\vartheta_{c,k})-\vartheta_{c,k}
            +2\sqrt{1+\vartheta_{c,k}}\Re\{\varphi_{c,k}^H\mathbf h_k^H\mathbf p_c  \} \\
            &-|\varphi_{c,k}|^2\left(|\mathbf h_k^H\mathbf p_c|^2+\sum\limits_{j=1}^K|\mathbf h_k^H\mathbf p_j|^2+\sigma_k^2\right),
        \end{aligned} \\
        &\begin{aligned}
            g_{p,k}&(\mathbf P,\vartheta_{p,k},\varphi_{p,k})\\
            &\triangleq \log(1+\vartheta_{p,k})-\vartheta_{p,k} +2\sqrt{1+\vartheta_{p,k}}\Re\{\varphi_{p,k}^H\mathbf h_k^H\mathbf p_k  \} \\
            &-|\varphi_{p,k}|^2\left(\sum\limits_{j=1}^K|\mathbf h_k^H\mathbf p_j|^2+\sigma_k^2\right).
        \end{aligned}
    \end{align}
\end{subequations}
Here, $\varphi_{c,k}$ and $\varphi_{p,k}$ are introduced as auxiliary variables to separate the fractional terms in $f_{c, k}$ and $f_{p, k}$.
These functions, $g_{c, k}$ and $g_{p, k}$, serve as lower bounds for $f_{c, k}$ and $f_{p, k}$, respectively, i.e.,
\begin{subequations}\label{g < f}
    \begin{align}
        g_{c, k} \leq f_{c, k}, \\
        g_{p, k} \leq f_{p, k}.
    \end{align}
\end{subequations}
The equalities are achieved when $\frac{\partial g_{c, k}}{\partial \varphi_{c, k}} = 0$ and $\frac{\partial g_{p, k}}{\partial \varphi_{p, k}} = 0$, leading to the following optimal values: 
\begin{subequations}\label{optimal phi}
    \begin{align}
        &\varphi_{c,k}^\circ = \frac{\sqrt{1+\vartheta_{c,k}}\mathbf{h}_{k}^{H}\mathbf p_c}{|\mathbf h_k^{H}\mathbf p_c|^2+\sum\limits_{j = 1}^K|\mathbf h_k^{H}\mathbf p_j|^2+\sigma_k^2}, \\
        &\varphi_{p,k}^\circ = \frac{\sqrt{1+\vartheta_{p,k}}\mathbf{h}_{k}^{H}\mathbf p_k}{\sum\limits_{j = 1}^K|\mathbf h_k^{H}\mathbf p_j|^2+\sigma_k^2}.
    \end{align}
\end{subequations}

From (\ref{f < r}) and (\ref{g < f}), we obtain that:
\begin{subequations}\label{relationship}
    \begin{align}
        &R_{c, k} = \max\limits_{\vartheta_{c, k}, \varphi_{c, k}}g_{c,k}(\mathbf P,\vartheta_{c,k},\varphi_{c,k}), \\
        &R_{p, k} = \max\limits_{\vartheta_{p, k}, \varphi_{p, k}}g_{p,k}(\mathbf P,\vartheta_{p,k},\varphi_{p,k}).
    \end{align}
\end{subequations}

By substituting $R_{c, k}$ and $R_{p, k}$ in problem (\ref{P1}) by (\ref{relationship}), problem (\ref{P1}) is equivalently transformed into:
\begin{subequations}\label{P2}
    \begin{align}
        \max_{\mathbf{P}, \mathbf{c}, \boldsymbol{\vartheta}_c, \boldsymbol{\vartheta}_p, \boldsymbol{\varphi}_c, \boldsymbol{\varphi}_p} &\min\limits_{k\in\mathcal{K}}\{g_{p, k}(\mathbf P,\vartheta_{p,k},\varphi_{p,k})+c_k\} \\
        \mathrm{s.t.}\ & \mathbf{1}^T\mathbf{c} \leq \min\limits_{k\in\mathcal{K}}\{g_{c, k}(\mathbf P,\vartheta_{c,k},\varphi_{c,k})\}, \\
        &\mathbf{c} \succeq \mathbf{0}, \\
        &\mathrm{tr}\left(\mathbf{P}^H\mathbf{P}\right) \leq P_t. 
    \end{align}
\end{subequations}
Although problem (\ref{P2}) remains non-convex, it is block-wise convex among variable blocks: $\{\mathbf{P}, \mathbf{c}\}$, $\{\boldsymbol{\vartheta}_c, \boldsymbol{\vartheta}_p\}$, $\{\boldsymbol{\varphi}_c, \boldsymbol{\varphi}_p\}$. Therefore, an alternating optimization (AO) framework can be employed to tackle this problem.
Each optimization block is optimized iteratively while fixing the other two variable blocks at their most recent values.
More precisely, we first find the optimal values of $\boldsymbol{\vartheta}_c$, $\boldsymbol{\vartheta}_p$, $\boldsymbol{\varphi}_c$, and $\boldsymbol{\varphi}_p$ using equations (\ref{optimal theta}) and (\ref{optimal phi}) with fixed $\mathbf{P}$ and $\mathbf{c}$.
Then, we utilize these fixed values to address problem (\ref{P2}).
Once $\boldsymbol{\vartheta}_c$, $\boldsymbol{\vartheta}_p$, $\boldsymbol{\varphi}_c$, and $\boldsymbol{\varphi}_p$ are fixed,
(\ref{P2}) transforms into a convex problem, which is solved using existing solvers in optimization toolboxes such as CVX \cite{grant2008cvx}. 
This is the standard FP algorithm that has been used in many existing works, i.e., in \cite{li2023sum} to address the related RSMA resource optimization problems.
Readers are referred to \cite{FPpartone, FPparttwo} for more details of its convergence analysis.
\par
While the standard FP algorithm can be applied to solve problem (\ref{P2}),
the iterative process of solving a sequence of convex subproblems using the interior-point method—typically implemented via optimization toolboxes such as CVX—leads to substantially increased computational time.
In the subsequent subsections,
our goal is to resolve problem (\ref{P2}) (with fixed $\boldsymbol{\vartheta}_c$, $\boldsymbol{\vartheta}_p$, $\boldsymbol{\varphi}_c$,
and $\boldsymbol{\varphi}_p$) by applying the extragradient method to its Lagrangian dual.
This approach eliminates the need for conventional interior-point method to solve (15) directly,
thereby removing dependence on standard optimization toolboxes like CVX.
Our method not only improves computational efficiency but also effectively addresses the non-convex and non-smooth nature of the MMF problem in RSMA.

\subsection{Lagrangian Dual Problem}
In this subsection, we focus on (\ref{P2}) with optimization variables $\mathbf{P}, \mathbf{c}$, and resort to its Lagrangian dual problem.
Subsequent subsections delve into a detailed exploitation of the Lagrangian dual problem and its solution.
\par 
With fixed values for $\boldsymbol{\vartheta}_c$, $\boldsymbol{\vartheta}_p$, $\boldsymbol{\varphi}_c$,  $\boldsymbol{\varphi}_p$,  and an auxiliary variable $t$ to represent the objective function, problem (\ref{P2})
is reformulated as:
\begin{subequations}
    \label{P3}
    \begin{align}
       {\mathcal{P}(\bm{\vartheta}, \bm{\varphi}):}\ \max\limits_{\mathbf P,\mathbf c,t}\,\, &t\\
        \text{s.t.}\,\,		&t\leq c_k+g_{p, k}(\mathbf{P}) \,\,\,\forall k\in\mathcal{K}, \label{eq:c1}\\
        &\mathbf{1}^T\mathbf{c} \leq g_{c,k}(\mathbf{P}),\,\,\, \forall k\in\mathcal{K}, \label{eq:c2}\\
        &\mathbf{c} \succeq \mathbf{0}, \label{eq:c3}\\
        &\mathrm{tr}(\mathbf{P}^H\mathbf{P}) \leq P_{t} \label{eq:c4}.
    \end{align}
\end{subequations}
Despite the convex nature of problem (\ref{P3}), attempting to solve it directly based on approaches like gradient descent involves a projection operation. However,  deriving a closed-form expression for this projection is rather challenging, leading to a high computational complexity to solve problem (\ref{P3}) directly.
To overcome this obstacle, we turn our focus to the Lagrangian dual problem, which allows us to circumvent the intricate projection operations.
\par 
By introducing the Lagrangian dual variables $\bm\lambda=[\lambda_1,\cdots,\lambda_K]^T$, $\bm\rho=[\rho_1,\cdots,\rho_K]^T$, $\bm\mu=[\mu_1,\cdots,\mu_K]^T$, and $ \omega $ respectively for constraints (\ref{eq:c1})--(\ref{eq:c4}),
the Lagrangian function for problem (\ref{P3}) is defined as:
\begin{equation}\label{Lagrangian function}
    \begin{aligned}
        L(\mathbf P,\mathbf c,t,\bm\lambda, \bm\rho,\bm\mu,\omega)&\triangleq t-\sum\limits_{k=1}^K \lambda_k\left(t-(c_k+g_{p, k}(\mathbf P))\right) \\
        &-\sum\limits_{k=1}^K\rho_k\left(\sum_{j=1}^K c_j-g_{c,k}(\mathbf P)\right)\\
        &+\sum_{k=1}^K\mu_kc_k-\omega(\mathrm{tr}(\mathbf P^H\mathbf P)-P_t).\\
    \end{aligned}
\end{equation}
Given that problem (\ref{P3}) is convex and satisfies Slater's condition,
the strong duality holds \cite{convex}. 
We choose to address the Lagrangian dual problem associated with (\ref{P3}) instead of directly solving problem (\ref{P3}).
The Lagrangian dual problem of (\ref{P3}) is formulated as:
\begin{subequations}
    \label{P4}
    \begin{align}
        \min\limits_{\bm\lambda, \bm\rho,\bm\mu,\omega}&\max\limits_{\mathbf P,\mathbf c,t}\ L(\mathbf P,\mathbf c,t,\bm\lambda, \bm\rho,\bm\mu,\omega) \\
        \text{s.t.}\ \ &\bm\lambda \succeq \mathbf{0}, \bm\rho \succeq \mathbf{0}, \bm\mu \succeq \mathbf{0}, \omega \geq 0. \label{Lagrangian constraints}
    \end{align}
\end{subequations}

\subsection{Variational Inequality Theory}
In this subsection, we introduce the variational inequality theory \cite{zhang2023ultra} and establish its connection to problem (\ref{P4}).
Leveraging this connection, we then solve problem (\ref{P4}) by using the extragradient method detailed in the next subsection.
\par 
\subsubsection{Preliminaries}
Variational inequality theory spans various fields,
finding applications in industry, finance, economics,
and both pure and applied sciences \cite{sedlmayer2023fast}.
A general form of a variational inequality is expressed as follows:
\begin{equation}
    \mathbf{h}(\mathbf{x})^T(\mathbf{x}' - \mathbf{x}) \geq 0,\ \forall \mathbf{x}' \in \mathcal{S}, \label{general VI}
\end{equation}
where $\mathcal{S} \subseteq \mathbb{R}^n$ denotes a closed convex set,
$\mathbf{h}(\cdot)$ is a mapping from $\mathcal{S}$ to $\mathbb{R}^n$,
and $\mathbf{x}$ is the vector to be determined.
This theory stands as a robust framework for solving problems that aim to find a vector satisfying the conditions of (\ref{general VI}).
It is helpful in various domains, especially in optimization.
For instance, it can help determine primal-dual pairs of optimal solutions in constrained convex optimization problems.

\par 
\subsubsection{Connection to problem (\ref{P4})}
Our objective is to maximize the Lagrangian function with respect to $(\mathbf{P},\mathbf{c},t)$ and simultaneously minimize it with respect to  $(\bm\lambda, \bm\rho,\bm\mu,\omega)$.
Therefore, solving problem (\ref{P4}) is equivalent to identifying the saddle point of the Lagrangian function due to the strong duality.
As problem (\ref{P4}) involves both complex and real variables,
we employ a decomposition strategy by splitting complex variables into their real and imaginary parts,
so as to align with the variational inequality theory.
Let us define:
\begin{equation}
\label{eq:yz}
    \mathbf{y} \triangleq
    \left[\begin{array}{c}
        \Re\{\mathbf{p}_c\} \\
        \Im\{\mathbf{p}_c\} \\
        \Re\{\mathbf{p}_1\} \\
        \vdots \\
        \Re\{\mathbf{p}_K\} \\
        \Im\{\mathbf{p}_1\} \\
        \vdots \\
        \Im\{\mathbf{p}_K\} \\
        \mathbf{c} \\
        t
    \end{array}\right],\ \ \ 
    \mathbf{z} \triangleq
    \left[\begin{array}{c}
        \bm\lambda \\
        \bm\rho \\
        \bm\mu \\
        \omega
    \end{array}\right].
\end{equation}
Given the concavity of $L(\mathbf{P},\mathbf{c},t,\bm\lambda, \bm\rho,\bm\mu,\omega)$ in $(\mathbf{P},\mathbf{c},t)$ and its convexity in $(\bm\lambda, \bm\rho,\bm\mu,\omega)$,
the optimality conditions indicate that the optimal solutions of $\mathbf{y}^\circ$ and $\mathbf{z}^\circ$ of problem (\ref{P4}) should satisfy the following inequalities:
\begin{subequations}
    \begin{align}
        &\left(\left.\frac{\partial L}{\partial\mathbf{y}}\right|_{\mathbf{y} = \mathbf{y}^\circ}\right)^T(\mathbf{y}' - \mathbf{y}^\circ) \leq 0, \\
        &\left(\left.\frac{\partial L}{\partial\mathbf{z}}\right|_{\mathbf{z} = \mathbf{z}^\circ}\right)^T(\mathbf{z}' - \mathbf{z}^\circ) \geq 0.
    \end{align}
\end{subequations}
where $\mathbf{y}'$ and $\mathbf{z}'$ denote any other points in the feasible set of problem (\ref{P4}).

\par
With this observation, we can then show that the optimal solution to problem (\ref{P4}) satisfies the variational inequality (\ref{general VI}).
For clarity, by further defining:
\begin{subequations}
\label{eq:xy}
\begin{align}
   & \mathbf{x} \triangleq [\mathbf{y}^T, \mathbf{z}^T]^T,\ \ \ \label{x definition}\\
    & \mathbf{h}(\mathbf{x}) \triangleq \left[-\left(\frac{\partial L}{\partial\mathbf{y}}\right)^T,
    \left(\frac{\partial L}{\partial\mathbf{z}}\right)^T\right]^T, \label{h definition}
    \end{align}
\end{subequations}
problem (\ref{P4}) is then equivalently formulated as  the following variational inequality  problem:
\begin{subequations}
    \label{variatinal inequality}
    \begin{align}
        \text{Find}\ & \mathbf{x} \in \mathcal{S}\\
        \text{s.t.}\ \ &\mathbf{h}(\mathbf{x})^T(\mathbf{x}' - \mathbf{x}) \geq 0,\ \forall \mathbf{x}' \in \mathcal{S},
    \end{align}
\end{subequations}
where $\mathbf{x}$ and $\mathbf{h}(\mathbf{x})$ are defined in (\ref{eq:xy}).
$\mathcal{S}$ is the feasible set of the Lagrangian dual problem (\ref{P4}), defined as:
\begin{equation}
  \mathcal{S} \triangleq \{\mathbf{x}\ |\ \mathbf{P}\in\mathbb{C}^{N_t\times K},\ \mathbf{c}\in \mathbb{R}^K,\ t\in \mathbb{R},\ \mathbf{z} \succeq \mathbf{0}\}. 
\end{equation}

\subsection{Extragradient Method}
\label{sec:EG}
Building upon the analysis in the previous two subsections, it is evident that resolving each subproblem (\ref{P3}) is equivalent to solving the variational inequality problem (\ref{variatinal inequality}).
In this subsection, our target is to tackle the variational inequality problem (\ref{variatinal inequality}) using the extragradient method.
\par
A mapping $\mathbf{h}(\cdot)$ is termed monotone on $\mathcal{S}$ if $\left(\mathbf{h}(\mathbf{x}) - \mathbf{h}(\mathbf{x}')\right)^T(\mathbf{x} - \mathbf{x}') \geq 0$ for all $\mathbf{x}, \mathbf{x}' \in \mathcal{S}$
(strongly monotone when $\left(\mathbf{h}(\mathbf{x}) - \mathbf{h}(\mathbf{x}')\right)^T(\mathbf{x} - \mathbf{x}') > 0$).
Since the Lagrangian dual function $L(\mathbf{P},\mathbf{c},t,\bm\lambda, \bm\rho,\bm\mu,\omega)$ defined in (\ref{Lagrangian function}) is concave in $(\mathbf{P},\mathbf{c},t)$ and convex in $(\bm\lambda, \bm\rho,\bm\mu,\omega)$,
the mapping $\mathbf{h}(\mathbf{x})$ defined in (\ref{h definition}) is monotone.
As indicated in \cite{nonconvex},
when a mapping $\mathbf{h}(\cdot)$ is monotone,
extragradient is a well-known approach for solving the variational inequality problem based on this mapping.
This inspires us to use extragradient to solve problem (\ref{variatinal inequality}).
In the following, we provide a detailed description of how to apply the extragradient algorithm to solve problem  (\ref{variatinal inequality}).
\par 
\subsubsection{Algorithm framework}
The extragradient algorithm  typically involves two iterative steps, necessitated by the alternating process of estimating and refining the solution:
\begin{itemize}
    \item \textit{First iterative step--Prediction}: In this step, the algorithm generates an initial estimation of the solution based on the current solution.
    \item \textit{Second iterative step--Correction}: This step involves correcting the prediction of the solution obtained in the first step, so as to ensure the solution moves towards the optimal solution.
\end{itemize}
By alternating between these two steps, the extragradient algorithm iteratively improves its estimate of the solution until convergence is attained. This two-step approach effectively balances exploration (exploring new solutions) with exploitation (refining promising solutions),  facilitating a gradual convergence towards an optimal solution while ensuring stability and efficiency throughout the optimization process.
    
The iterative steps of the extragradient method for finding the vector $\mathbf{x}$ satisfying (\ref{variatinal inequality})  are outlined as follows.
Assuming that  $\mathbf{x}^{(n)}$ is already obtained, then at iteration $n + 1$, $\mathbf{x}^{(n+1)}$ is updated as:
\begin{subequations}\label{extragradient}
    \begin{align}
    \textbf{Step 1: } & \textbf{Prediction} \nonumber\\
        &\bar{\mathbf{x}}^{(n)} = \operatorname{Proj}_{\mathcal{S}}\left(\mathbf{x}^{(n)} - \alpha^{(n)} \mathbf{h}\left(\mathbf{x}^{(n)}\right)\right), \label{prediction}  \\
    \textbf{Step 2: } & \textbf{Correction} \nonumber\\
        &\mathbf{x}^{(n+1)} = \operatorname{Proj}_{\mathcal{S}}\left(\mathbf{x}^{(n)} - \alpha^{(n)} \mathbf{h}\left(\bar{\mathbf{x}}^{(n)}\right)\right), \label{correction}
    \end{align}
\end{subequations}
where $\bar{\mathbf{x}}^{(n)}$ serves as a provisional solution intended for the projection described in (\ref{correction}),
$\alpha^{(n)}$ denotes the step size,
and $\operatorname{Proj}_{\mathcal{S}}(\cdot)$ is the projection onto the set $\mathcal{S}$.
Specifically, the projection operator $\operatorname{Proj}_{\mathcal{S}}(\cdot)$ is equivalent to $\max (0, \cdot)$ in this paper. 
\par
 It is worth noting that the extragradient algorithm is specifically designed to solve variational inequalities,
 accommodating cases where $\mathbf{h}(\cdot)$ lacks strong monotonicity.
 In contrast, the classical one-step gradient algorithm can only handle strongly monotone problems.
 As our problem involves bilinear terms, such as $\sum_{k=1}^K\mu_kc_k$, which are not strongly monotone,
 the extragradient algorithm is a more fitting choice \cite{marcotte1991application}.  
The extragradient algorithm specified in (25) incorporates the idea of extrapolation from the general gradient method to adjust gradients, in order to accelerate convergence and improve the performance of the general gradient algorithm.
The ``extra'' in the extragradient algorithm refers to the additional step (correction
step) that corrects the iteration direction to ensure convergence.

It is important to highlight that relying solely on executing procedure (\ref{prediction}) can result in $\bar{\mathbf{x}}^{(n)}$ being distant from the optimal solution.
Therefore, procedure (\ref{correction}) is utilized to refine the solution and bring it closer to optimality.
To further clarify this concept, we use a basic min-max problem $\min_{x}\max_{y}\ (xy)$  as an example.
In this scenario, $(x, y)=(0, 0)$ is a unique saddle point.
At any given point $(x^{(n)}, y^{(n)})$,
the direction $\mathbf{h}(x^{(n)}, y^{(n)}) = (y^{(n)}, -x^{(n)})$ stands orthogonal to $(x^{(n)}, y^{(n)})$.
Thus,  executing  (\ref{prediction}) only increases the distance from the saddle point.
If we further employ (\ref{correction}), the direction $-\mathbf{h}(\bar{x}^{(n)}, \bar{y}^{(n)}) = (-\bar{y}^{(n)}, \bar{x}^{(n)})$ tends to approach the saddle point,
where $\bar{x}^{(n)} = {x}^{(n)} - \alpha^{(n)}y^{(n)}$ and $\bar{y}^{(n)} = {y}^{(n)} + \alpha^{(n)}x^{(n)}$.
In Fig. \ref{divergence}, we respectively illustrate the convergence behaviors of using the classical one-step gradient algorithm and the two-step extragradient algorithm to solve the basic min-max problem $\min_{x}\max_{y}\ (xy)$. It is evident that with increasing iterations, only the two-step extragradient algorithm eventually converges to the optimal solution $(0,0)$. Therefore, (\ref{correction}) is essential and necessary.

\begin{figure}[htbp]
    \centering
    \subfloat[One-step gradient algorithm]{\includegraphics[width=0.24\textwidth]{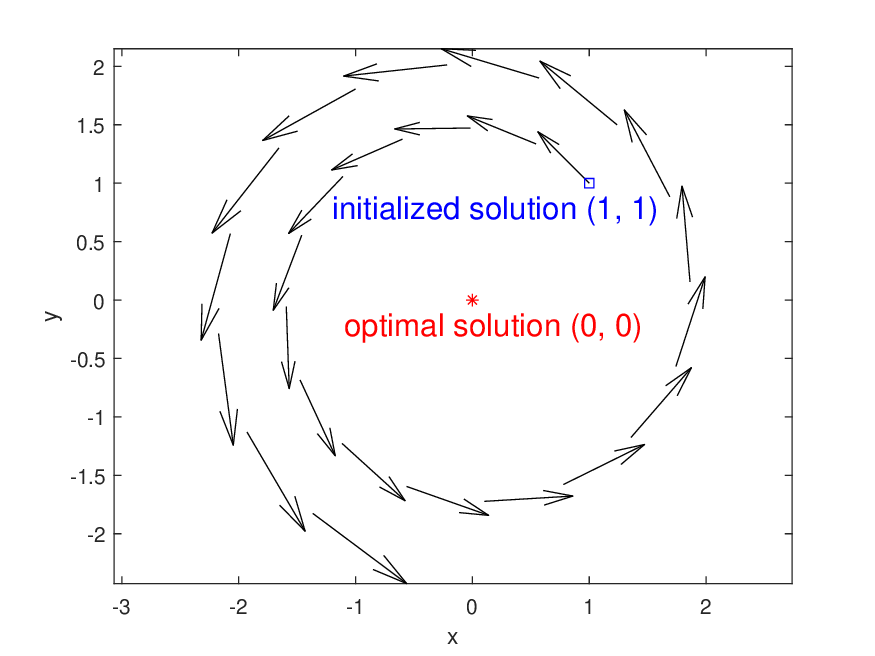}}
    \subfloat[Extragradient algorithm]{\includegraphics[width=0.24\textwidth]{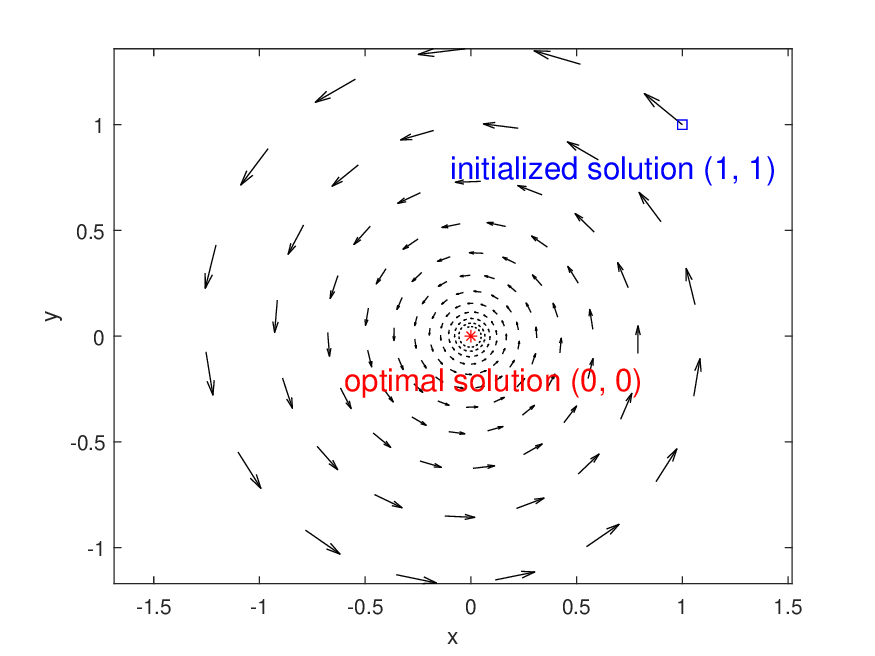}}

    \caption{{Comparison between the optimization trajectories of the one-step gradient algorithm and the extragradient algorithm for solving the toy example problem $\min_{x}\max_{y}\ (xy)$.  The arrow indicates the direction of the solution point as the iteration progresses.}}
    \label{divergence}
\end{figure}

\par 
\subsubsection{Step size design}
\label{step size design}
If the mapping $\mathbf{h}(\cdot)$ is monotone and $L$-Lipschitz continuous,
the extragradient algorithm in (\ref{extragradient}) is guaranteed to converge to the optimal solution to problem  (\ref{variatinal inequality}) if the step size is bounded by $\frac{1}{L}$ \cite{zhang2023ultra}.
However, finding the Lipschitz constant $L$ for the mapping $\mathbf{h}(\cdot)$ in (\ref{variatinal inequality}) is rather challenging in our problem.
To tackle this challenge, we take advantage of  Khobotov's theorem\cite{marcotte1991application} as specified in the following Theorem \ref{Khobotov theorem}.
\begin{theorem}
\label{Khobotov theorem}
    (Khobotov's theorem) Let $\mathcal{S}$ be a closed convex subset of 
    $\mathbb{R}^n$ and $\mathbf{h}(\cdot)$ be a monotone mapping from $\mathcal{S}$ into $\mathbb{R}^n$.
    Let $\mathbf{x}^{(0)}\in \mathcal{S}$.
    $\{\mathbf{x}^{(n)}\}_{n\geq 1}$ and $\{\bar{\mathbf{x}}^{(n)}\}_{n\geq 1}$ are the sequences defined by (\ref{extragradient}).
    Then we have, for any nonnegative sequence ${\alpha^{(n)}}$ and $\mathbf{x}^\circ \in \mathcal{S}$:
    \begin{equation}
    \begin{aligned}
        &\| \mathbf{x}^{(n+1)} - \mathbf{x}^\circ\|^2 \\
        \leq &\| \mathbf{x}^{(n)} - \mathbf{x}^\circ\|^2 \\
        -&\| \mathbf{x}^{(n)} - \bar{\mathbf{x}}^{(n)}\|^2
        \left(1 - \left(\alpha^{(n)}\right)^2\frac{\|h(\mathbf{x}^{(n)}) - h(\bar{\mathbf x}^{n})\|^2}{\|\mathbf{x}^{(n)} - \bar{\mathbf{x}}^{(n)}\|^2}\right),
    \end{aligned}
\end{equation}
\end{theorem}
\begin{proof}
    The theorem can be proved using the triangle inequality of $l_2$-norm.
    Readers are referred to \cite{marcotte1991application} for more details.
\end{proof}

Khobotov's theorem provides us with a powerful tool to design the step size.
In order to take advantage of Khobotov's theorem, we set $\mathbf{x}^\circ$ to be the optimal solution to problem (\ref{variatinal inequality}), which satisfies:
\begin{equation}
    \mathbf{h}(\mathbf{x}^\circ)^T(\mathbf{x}' - \mathbf{x}^\circ) \geq 0,\ \forall \mathbf{x}' \in \mathcal{S}.
\end{equation}
Khobotov's theorem reveals that we have to adjust the step size $\alpha^{n}$ to make sure $\| \mathbf{x}^{(n+1)} - \mathbf{x}^\circ\| < \| \mathbf{x}^{(n)} - \mathbf{x}^\circ\|$.
By setting $\alpha^{(n)} = \beta \frac{\|\mathbf{x}^{(n)} - \bar{\mathbf{x}}^{(n)}\|}{\|h(\mathbf{x}^{(n)}) - h(\bar{\mathbf x}^{(n)})\|}$,
where $\beta \in (0, 1)$ is a constant, then we have:
\begin{equation}
    \begin{aligned}
        &\| \mathbf{x}^{(n+1)} - \mathbf{x}^\circ\|^2 \\
        \leq &\| \mathbf{x}^{(n)} - \mathbf{x}^\circ\|^2 - \| \mathbf{x}^{(n)} - \bar{\mathbf{x}}^{(n)}\|^2\left(1 - \beta^2\right) \\
        < &\| \mathbf{x}^{(n)} - \mathbf{x}^\circ\|^2.
    \end{aligned}
\end{equation}
This implies that $\mathbf{x}^{(n+1)}$ gets closer to the optimal solution than $\mathbf{x}^{(n)}$. Therefore, the sequence $\left\{\mathbf{x}^{(n)}\right\}$ converges to the optimal solution to problem (\ref{P4}) as $n$ increases.
\par
Inspired by the discovery from Khobotov's theorem, we propose to update the step size of the extragradient method using the following procedure.
In each iteration $n+1$, we initially derive $\bar{\mathbf{x}}^{(n)}$ from (\ref{prediction}) with a initialized $\alpha$.
Then, we compute the step size $\alpha' = \beta \frac{\|\mathbf{x}^{(n)} - \bar{\mathbf{x}}^{(n)}\|}{\|h(\mathbf{x}^{(n)}) - h(\bar{\mathbf x}^{(n)})\|}$.
After that we set $\alpha^{(n)} = \min\{\alpha', \alpha\}$ and use $\alpha^{(n)}$ for iteration $n+1$ to update $\mathbf{x}^{(n+1)}$ from the two steps in (\ref{extragradient}).

\par 
\subsubsection{Gradient and variable updates}
We have already introduced the general optimization framework of the extragradient algorithm in Section \ref{sec:EG}-1 and the design of the step size in Section \ref{sec:EG}-2.
In this part, we provide a comprehensive implementation procedure of applying the extragradient algorithm to solve problem (\ref{variatinal inequality}),
 which is an equivalent transformation of the Lagrangian dual problem (\ref{P4}).

A critical step in the extragradient algorithm outlined in (\ref{extragradient}) is computing the monotone mapping $\mathbf{h}(\mathbf{x})$.
According to its definition in (\ref{eq:xy}),
the update of $\mathbf{h}(\mathbf{x})$ involves calculating the derivatives of the Lagrangian function $L(\mathbf{P},\mathbf{c},t,\bm{\lambda}, \bm{\rho},\bm{\mu},\omega)$ defined in (17).
Next, we specify the derivatives of the Lagrangian function (\ref{Lagrangian function}) with respect to the complex variables $\mathbf{P}$ and real variables $\mathbf{c},t,\bm{\lambda}, \bm{\rho},\bm{\mu},\omega$,
as well as the procedure of updating the corresponding variables using (\ref{extragradient}).

First, we show how to use (\ref{extragradient}) to update complex beamforming variables $\mathbf{P}$.
Suppose that $f(z)$ is a complex function, where $z = x + jy \in \mathbb{C}$, $x = 
\Re (z)$ and $y = \Im (z)$.
The partial derivative  $\frac{\partial f(z)}{\partial z^*}$ with respective to the complex conjugate $z^*$ is defined as $\frac{\partial f(z)}{\partial z^*} \triangleq \frac{1}{2}\left(\frac{\partial f(z)}{\partial x} + j\frac{\partial f(z)}{\partial y}\right)$.
Readers are referred to \cite{hunger631019introduction} for more details of Wirtinger calculus.
The partial derivatives of the beamforming vectors are therefore given by:
\begin{subequations}\label{P gradient}
    \begin{align}
        &\frac{\partial L}{\partial \mathbf{p}_c^*} = \sum_{j = 1}^K\rho_j \left(\varphi_{c,j}\sqrt{1+\vartheta_{c, j}}\mathbf{h}_j - |\varphi_{c, j}|^2\mathbf{h}_j\mathbf{h}_j^H\mathbf{p}_c\right) - \omega \mathbf{p}_c, \\
        &\begin{aligned}
            \frac{\partial L}{\partial \mathbf{p}_k^*} &= \lambda_k \varphi_{p, k}\sqrt{1+\vartheta_{p, k}}\mathbf{h}_k \\
            &- \sum_{j = 1}^K(\lambda_j|\varphi_{p, j}|^2 + \rho_j|\varphi_{c,j}|^2)\mathbf{h}_j\mathbf{h}_j^H\mathbf{p}_k - \omega \mathbf{p}_k.
        \end{aligned}
    \end{align}
\end{subequations}
Based on (\ref{P gradient}), the update procedures of the beamforming vectors,
following the extragradient algorithm described in (\ref{extragradient}),
are given by:
\begin{subequations}\label{extragradient P}
    \begin{align}
        &\bar{\mathbf{p}}_c^{(n)} = \mathbf{p}_c^{(n)} + 2\alpha^{(n)} \left.\frac{\partial L}{\partial \mathbf{p}_c^*}\right|_{\mathbf{x} = \mathbf{x}^{(n)}}, \label{prediction p_c} \\
        &\bar{\mathbf{p}}_k^{(n)} = \mathbf{p}_k^{(n)} + 2\alpha^{(n)} \left.\frac{\partial L}{\partial \mathbf{p}_k^*}\right|_{\mathbf{x} = \mathbf{x}^{(n)}}, \label{prediction p_k} \\
        &\mathbf{p}_c^{(n+1)} = \mathbf{p}_c^{(n)} + 2\alpha^{(n)} \left.\frac{\partial L}{\partial \mathbf{p}_c^*}\right|_{\mathbf{x} = \bar{\mathbf{x}}^{(n)}}, \label{correction p_c} \\
        &\mathbf{p}_k^{(n+1)} = \mathbf{p}_k^{(n)} + 2\alpha^{(n)} \left.\frac{\partial L}{\partial \mathbf{p}_k^*}\right|_{\mathbf{x} = \bar{\mathbf{x}}^{(n)}}. \label{correction p_k}
    \end{align}
\end{subequations}
where $\mathbf{x}$ is defined in (\ref{x definition}),
and it contains all variables.
$\mathbf{x}^{(n)}$ represents the solution point updated at the $n$th iteration.

Similarly, with the gradients of the real common rate allocation variable $c_k$ and the objective function  variable $t$ respectively given as $\frac{\partial L}{\partial c_k} = \lambda_k + \mu_k - \sum_{j = 1}^K \rho_j$ and $\frac{\partial L}{\partial t} = 1 - \sum_{j = 1}^K\lambda_j$,
the update procedures for $c_k$ and $t$ are as follows:
\begin{subequations}\label{extragradient c t}
    \begin{align}
        &\bar{c}_k^{(n)} = c_k^{(n)} + \alpha^{(n)} \left.\frac{\partial L}{\partial c_k}\right|_{\mathbf{x} = \mathbf{x}^{(n)}}, \label{prediction c}\\
        &\bar{t}^{(n)} = t^{(n)} + \alpha^{(n)} \left.\frac{\partial L}{\partial t}\right|_{\mathbf{x} = \mathbf{x}^{(n)}}, \label{prediction t}\\
        &c_k^{(n+1)} = c_k^{(n)} + \alpha^{(n)} \left.\frac{\partial L}{\partial c_k}\right|_{\mathbf{x} = \bar{\mathbf{x}}^{(n)}}, \label{correction c}\\
        &t^{(n+1)} = t^{(n)} + \alpha^{(n)} \left.\frac{\partial L}{\partial t}\right|_{\mathbf{x} = \bar{\mathbf{x}}^{(n)}}. \label{correction t}
    \end{align}
\end{subequations}
It is worth noting that we focus on addressing the Lagrangian dual problem (\ref{P4}) using the proposed extragradient algorithm.
The objective function of (\ref{P4}) is penalized with a weighted sum of the constraint functions and no constraints are imposed on the primal variables $\mathbf P,\mathbf c,t$. 
Therefore, no projection operator is applied in (\ref{extragradient P}) and (\ref{extragradient c t}).

Following the same procedure, the update of the Lagrangian multipliers $\mathbf{z} = [\bm\lambda^T, \bm\rho^T,\bm\mu^T,\omega]^T$ are given by:
\begin{subequations}\label{Lagrangian multipliers update}
    \begin{align}
        &\bar{\mathbf z}^{(n)} = \left(\mathbf{z}^{(n)} - \alpha^{(n)} \left.\frac{\partial L}{\partial \mathbf{z}}\right|_{\mathbf{x} = \mathbf{x}^{(n)}}\right)^+, \label{prediction z}\\
        &\mathbf{z}^{(n+1)} = \left(\mathbf{z}^{(n)} - \alpha^{(n)} \left.\frac{\partial L}{\partial \mathbf{z}}\right|_{\mathbf{x} = \bar{\mathbf{x}}^{(n)}}\right)^+, \label{correction z}
    \end{align}
\end{subequations}
where $\frac{\partial L}{\partial \mathbf{z}} = [(\frac{\partial L}{\partial \bm\lambda})^T, (\frac{\partial L}{\partial \bm\rho})^T, (\frac{\partial L}{\partial \bm\mu})^T, \frac{\partial L}{\partial \omega}]^T$.
The derivatives with respect to $\bm\lambda$, $\bm\rho$, $\bm\mu$ and $\omega$ can be calculated by $\frac{\partial L}{\partial \lambda_k} = -(t - c_k - g_{p, k})$, $\frac{\partial L}{\partial \rho_k} = -\left(\sum_{j = 1}^Kc_j - g_{c, k}\right)$,
$\frac{\partial L}{\partial \mu_k} = c_k$ and $\frac{\partial L}{\partial \omega} = -\left(\mathrm{tr}(\mathbf{P}^H\mathbf{P}) - P_t\right)$, respectively.
Since all Lagrangian multipliers are non-negative,
the projection operator $\operatorname{Proj}_{\mathcal{S}}(\cdot)$ in (\ref{extragradient}) is equivalent to $(\cdot)^+ \triangleq \max(0, \cdot)$ here.

\subsubsection{Proposed extragradient-based FP (EG-FP) algorithm} 
\par
Combining the extragradient algorithm with the FP algorithm in Section \ref{sec:FP}, we propose the extragradient-based FP (EG-FP) algorithm for addressing problem (\ref{P1}).
Fig. \ref{flowchart} shows the framework of the proposed EG-FP algorithm.
The detailed procedure is outlined in Algorithm \ref{extragradient-FP}.
Line 5 fixes $\boldsymbol{\vartheta}_c$, $\boldsymbol{\vartheta}_p$, $\boldsymbol{\varphi}_c$ and $\boldsymbol{\varphi}_p$ to obtain problem (\ref{P3}),
which is a subproblem of problem (\ref{P2}).
The loop running from lines 9--16 focuses on solving problem (\ref{P4}) or (\ref{variatinal inequality}) via the extragradient algorithm.
Notably, lines 10--12 determine the step size using the approach outlined in Section \ref{step size design}.
Line 14 calculates the MMF rate based on the beamforming matrix $\mathbf{P}^{(n+1)}$ for the stopping criterion of the inner loop.
Additionally, line 17 guarantees that the beamforming matrix satisfies the power constraint (\ref{P1 d}).
Lastly, line 18 computes the MMF rate based on the beamforming vectors $\mathbf{P}$ for the stopping criterion of the outer loop.

\begin{figure}
    \centering
    \includegraphics[width=\linewidth]{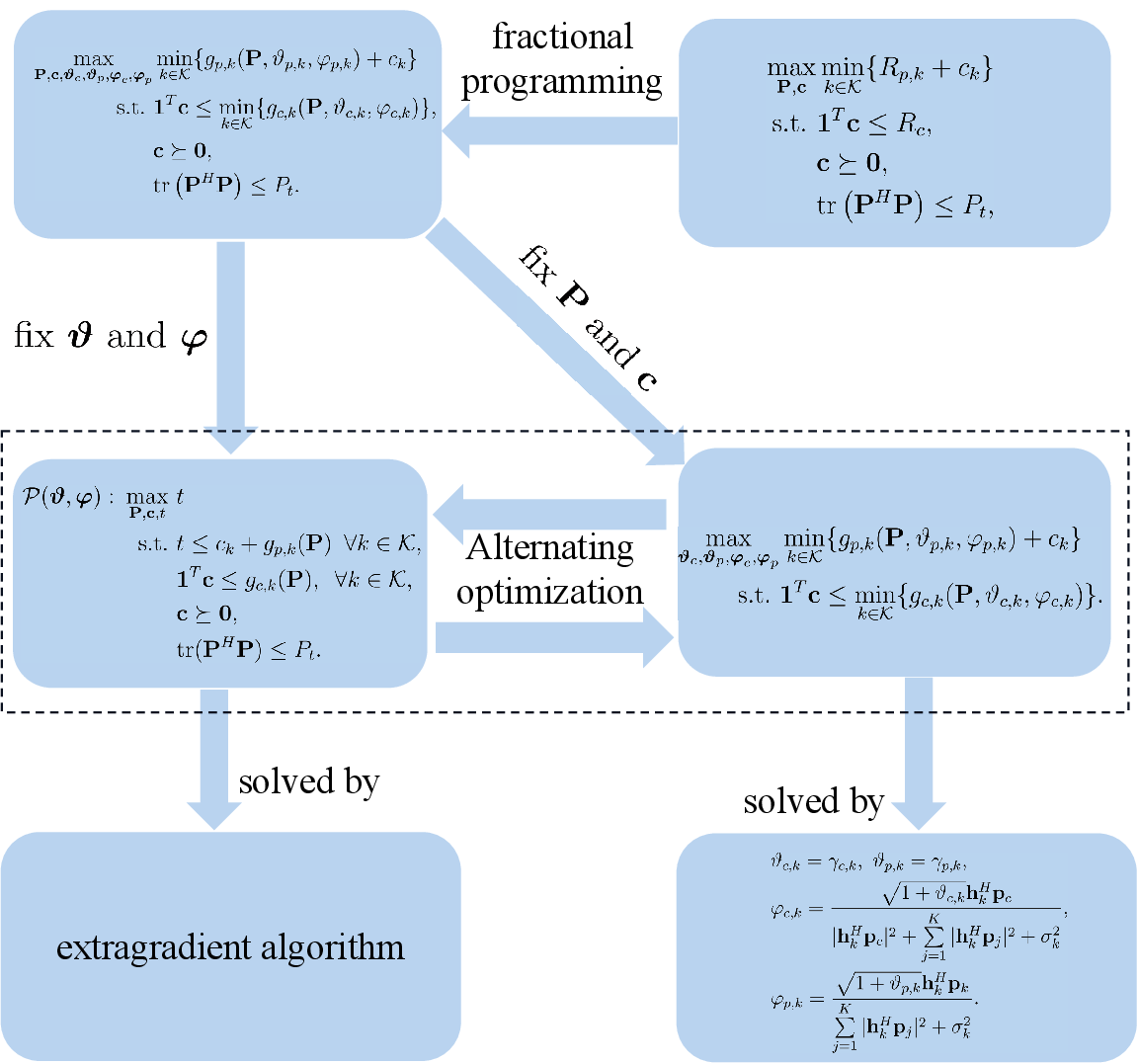}
    \caption{{Framework of the proposed EG-FP algorithm}}
    \label{flowchart}
\end{figure}

\begin{algorithm}[!t]
    \caption{The proposed extragradient-based FP (EG-FP) algorithm for solving problem (\ref{P1}).}
    \label{extragradient-FP}
    \textbf{Initialize:} $\mathbf P, \mathbf{c}, t, \alpha, \beta$;

    $outer\_obj^{(0)} = \text{MMF rate}(\mathbf{P})$.

    $m = 0$.

    \Repeat{$|outer\_obj^{(m)} - outer\_obj^{(m-1)}| < \epsilon_1$}{
        Update the auxiliary variables $\boldsymbol{\vartheta}_c$, $\boldsymbol{\vartheta}_p$, $\boldsymbol{\varphi}_c$ and $\boldsymbol{\varphi}_p$ by (\ref{optimal theta}) and (\ref{optimal phi}).

        $\mathbf{P}^{(0)} = \mathbf{P}$, $n = 0$.
        
        $inner\_obj^{(0)} = \text{MMF rate}(\mathbf{P}^{(0)})$.

        \textbf{Initialize:} $\mathbf{c}^{(0)}$, $t^{(0)}$, $\bm\lambda^{(0)}$, $\bm\rho^{(0)}$, $\bm\mu^{(0)}$, $\omega^{(0)}$.
        
        \Repeat{$|inner\_obj^{(n)} - inner\_obj^{(n-1)}| < \epsilon_2$}{
            Update $\bar{\mathbf{p}}_c^{(n)}$, $\bar{\mathbf{p}}_k^{(n)}$, $\bar{\mathbf{c}}^{(n)}$, $\bar{t}^{(n)}$, $\bar{\mathbf{z}}^{(n)}$ through
            (\ref{prediction p_c}), (\ref{prediction p_k}), (\ref{prediction c}), (\ref{prediction t}), (\ref{prediction z}) with $\alpha^{(n)} = \alpha$.

            Calculate $\alpha' = \beta \frac{\|\mathbf{x}^{(n)} - \bar{\mathbf{x}}^{(n)}\|}{\|h(\mathbf{x}^{(n)}) - h(\bar{\mathbf x}^{(n)})\|}$.

            $\alpha^{(n)} = \min\{\alpha', \alpha\}$.

            Update $\mathbf{p}_c^{(n+1)}$, $\mathbf{p}_k^{(n+1)}$, $\mathbf{c}^{(n+1)}$, $t^{(n+1)}$, $\mathbf{z}^{(n+1)}$ through (\ref{extragradient P}), (\ref{extragradient c t}), (\ref{Lagrangian multipliers update}) with $\alpha^{(n)}$.

            $inner\_obj^{(n+1)} = \text{MMF rate}(\mathbf{P}^{(n+1)})$.

            $n = n + 1$.
        }

        $\mathbf{P} = \sqrt{P_t}\frac{\mathbf{P}^{(n-1)}}{\| \mathbf{P}^{(n-1)} \|_F}$.
        
        $outer\_obj^{(m+1)} = \text{MMF rate}(\mathbf{P})$.
        
        $m = m + 1$.
    }
\end{algorithm}

\subsection{Low-dimensional Beamforming Design}\label{low dimensional form}
In this subsection, we provide a low-dimensional reformulation of problem (\ref{P1}), where the dimension of the beamforming matrix decreases from $N_t\times (K+1)$ to $K\times (K+1)$, independent of the number of transmit antennas.
By applying the proposed EG-FP algorithm to solve the reformulated problem, we propose a novel low-dimensional EG-FP algorithm, which significantly reduces the computational complexity. 
This low-dimensional EG-FP is particularly advantageous in scenarios where the number of transmit antenna $N_t$ is extremely large.
\par
\subsubsection{Optimal beamforming structure}
We first unveil the optimal beamforming structure for problem (\ref{P1}) based on the following Theorem \ref{prop:optBF}.

\begin{theorem}
\label{prop:optBF}
    Each optimal beamforming vector $\mathbf{p}_k$, $k\in\{c,1,2,\ldots,K\}$ of problem (\ref{P1}) is a linear combination of channel vectors $\mathbf{h}_k$, $\forall k\in\mathcal{K}$.
    In other words, each optimal beamforming vector $\mathbf{p}_k$ of problem (\ref{P1}) is to find a weight vector $\mathbf{q}_k \in \mathbb{C}^{K\times 1}$ such that
    \begin{equation}
        \mathbf{p}_k = \mathbf{H}\mathbf{q}_k, \forall k\in\{c,1,2,\ldots,K\},
    \end{equation}
    where $\mathbf{H} \triangleq [\mathbf{h}_1, \cdots, \mathbf{h}_K]$.
\end{theorem}

\begin{proof}
We start by examining the first-order stationary conditions of the KKT conditions for problem (\ref{P3}) given by:
\begin{equation}\label{KKT}
        \frac{\partial L}{\partial \mathbf{p}_c} = 0, \,\,
        \frac{\partial L}{\partial \mathbf{p}_k} = 0.
\end{equation}
By solving (\ref{KKT}), we derive the following optimal solution:
\begin{subequations}\label{linear combination}
    \begin{align}
        &\mathbf{p}_c =
        \frac{1}{\omega}\sum_{j = 1}^K\rho_j \left(\varphi_{c,j}\sqrt{1+\vartheta_{c, j}} - |\varphi_{c, j}|^2\mathbf{h}_j^H\mathbf{p}_c\right)\mathbf{h}_j, \\
        &\begin{aligned}
            \mathbf{p}_k &= \frac{1}{\omega}\lambda_k \varphi_{p, k}\sqrt{1+\vartheta_{p, k}}\mathbf{h}_k \\
            &- \frac{1}{\omega}\sum_{j = 1}^K\left(\lambda_j|\varphi_{p, j}|^2 + \rho_j|\varphi_{c,j}|^2\right)\left(\mathbf{h}_j^H\mathbf{p}_k\right)\mathbf{h}_j, \ \ \ \forall k \in \mathcal{K}.
        \end{aligned}
    \end{align}
\end{subequations}
Equations in (\ref{linear combination}) reveal that the optimal beamforming vectors for the convex subproblem (\ref{P3}) are, in fact, linear combinations of channel vectors.
\par 
Next, we show that the optimal beamforming vectors for the original non-convex non-smooth problem (\ref{P1}) can also be represented by linear combinations of channel vectors in (\ref{linear combination}).
Define $\mathbf{P}^\diamond$ as a locally optimal solution for problem (\ref{P2}).
Employing  this solution to compute (\ref{optimal theta}) and (\ref{optimal phi}) yields the optimal beamforming solution for problem $\mathcal{P}(\bm{\vartheta}(\mathbf{P}^\diamond), \bm{\varphi}(\mathbf{P}^\diamond))$, which is also $\mathbf{P}^\diamond$.
This indicates that any locally optimal beamforming vectors of problem (\ref{P2}) are linear combinations of channel vectors.
Since problem (\ref{P2}) is equivalent to problem (\ref{P1}), any locally optimal beamforming vectors of problem (\ref{P1}) are also linear combinations of channel vectors.
Moreover, as the globally optimal beamforming solution to problem (\ref{P1}) is one of the local-optimal solutions within the same problem (\ref{P1}),  the optimal beamforming vectors can be expressed as linear combinations of channel vectors. This completes our proof.
\end{proof}

\par
\subsubsection{Problem reformulation}
Based on the above analysis, we express $\mathbf{p}_c = \mathbf{H}\mathbf{q}_c$ and $\mathbf{p}_k = \mathbf{H}\mathbf{q}_k$, $\forall k \in \mathcal{K}$,
where $\mathbf{q}_c \in \mathbb{C}^{K\times 1}$ and $\mathbf{q}_k \in \mathbb{C}^{K\times 1}$.
With this transformation, the original problem (\ref{P1}) becomes:
\begin{subequations}\label{P5}
    \begin{align}
        &\max_{\mathbf Q, \mathbf c} \min\limits_{k\in\mathcal{K}}\left\{c_k+\log\left( 1 + \frac{|\mathbf g_k^H\mathbf q_k|^2}
        {\sum_{j\neq k}^K|\mathbf g_k^H\mathbf q_j|^2 + \sigma_k^2}\right)\right\} \\
        \mathrm{s.t.}\ & \mathbf{1}^T\mathbf{c} \leq \min_{k\in\mathcal{K}}\left\{\log\left(1+\frac{|\mathbf g_k^H\mathbf q_c|^2}
    {\sum_{j=1}^K|\mathbf g_k^H\mathbf q_j|^2 + \sigma_k^2}\right)\right\}, \\
        &\mathbf{c} \succeq \mathbf{0}, \\
        &\mathrm{tr}\left(\mathbf{Q}^H\mathbf{G}\mathbf{Q}\right) \leq P_t,
    \end{align}
\end{subequations}
where $\mathbf{G} = [\mathbf{g}_1, \cdots, \mathbf{g}_K] \triangleq \mathbf{H}^H\mathbf{H}$ and $\mathbf{Q} \triangleq [\mathbf{q}_c, \mathbf{q}_1, \cdots, \mathbf{q}_K]$.

It is important to note that the structure of problem (\ref{P5}) closely resembles that of the original problem (\ref{P1}).
However, in contrast to the original problem (\ref{P1}) with $N_t(K+1)+K$ variables, the transformed problem (\ref{P5}) has a significantly smaller dimension of optimization variables--specifically, $K(K+2)$ variables.
This dimension reduction is particularly appealing in massive MIMO systems where $N_t \gg  K$. 

\par
\subsubsection{Proposed low-dimensional EG-FP algorithm}
By exploiting the low-dimensional beamforming structure in the optimal solution (\ref{linear combination}), solving problem (\ref{P5}) instead of (\ref{P1}) can result in substantial computational savings, especially in massive MIMO scenarios. This motivates us to introduce the low-dimensional EG-FP algorithm, wherein we apply the EG-FP algorithm presented in Algorithm \ref{extragradient-FP}  to solve problem (\ref{P5}).
The steps of the proposed low-dimensional EG-FP algorithm are outlined below:
\begin{itemize}
    \item \textit{Step 1}: Transform problem (\ref{P1}) into (\ref{P5}).
    \item \textit{Step 2}: Apply FP  to decompose problem (\ref{P5}) into a series of convex subproblems for each variable block of $\{\mathbf{Q}, \mathbf{c}\}$, $\{\boldsymbol{\vartheta}_c, \boldsymbol{\vartheta}_p\}$, $\{\boldsymbol{\varphi}_c, \boldsymbol{\varphi}_p\}$.
    \item \textit{Step 3}: The subproblems of $\{\boldsymbol{\vartheta}_c, \boldsymbol{\vartheta}_p\}$, $\{\boldsymbol{\varphi}_c, \boldsymbol{\varphi}_p\}$ are solved in closed form as (\ref{optimal theta}) and (\ref{optimal phi}).  The subproblem of $\{\mathbf Q, \mathbf c\}$ is solved by the extragradient method. 
     \item \textit{Step 4}: Repeat Step 3 until convergence.
\end{itemize}
The algorithm framework for Algorithm \ref{extragradient-FP} can be easily extended to the low-dimensional EG-FP algorithm by respectively substituting  $\{\mathbf{h}_k\}$ and $\{\mathbf{p}_k\}$ in Algorithm \ref{extragradient-FP} with $\{\mathbf{g}_k\}$ and $\{\mathbf{q}_k\}$ as defined right below (\ref{P5}). The details are omitted here to eliminate redundancy.

\subsection{Convergence and Computational Complexity Analysis}

\textit{Convergence analysis}:
The step size design,
as outlined in Section \ref{step size design},
provides a brief insight into the convergence analysis.
By setting $\alpha^{(n)} = \beta \frac{\|\mathbf{x}^{(n)} - \bar{\mathbf{x}}^{(n)}\|}{\|h(\mathbf{x}^{(n)}) - h(\bar{\mathbf x}^{n})\|}$,
two contiguous iteration steps display a relationship where $\| \mathbf{x}^{(n+1)} - \mathbf{x}^\circ\| < \| \mathbf{x}^{(n)} - \mathbf{x}^\circ\|$.
This intuitively shows the convergence of the extragradient algorithm.
The detailed proof is provided in \cite{marcotte1991application} for interested readers.

As for the outer-loop convergence,
we provide a theorem here.
\begin{theorem}\label{outer convergence}
    If the inner-loop of Algorithm \ref{extragradient-FP} converges to the optimal solution of problem (\ref{P3}),
    the outer-loop will converge.
\end{theorem}

\begin{proof}
    See Appendix \ref{outer convergence proof}.
\end{proof}

\par
\textit{Computational analysis}: The computational complexities of both EG-FP and low-dimensional EG-FP are analyzed here.
For EG-FP, the predominant complexity within each iteration of Algorithm \ref{extragradient-FP} arises from the gradient calculation for $\mathbf{P}$,
scaling at $\mathcal{O}(K^2N_t)$.
Therefore, the overall complexity of Algorithm \ref{extragradient-FP} is $\mathcal{O}(\log(\epsilon_1^{-1})\sqrt{\epsilon_2^{-1}}K^2N_t)$,
where $\log(\epsilon_1^{-1})$ and $\sqrt{\epsilon_2^{-1}}$ correspond to the iteration counts for the outer and inner iteration loops, respectively.
The outer iteration count is derived from \cite{liping2020}, while the inner iteration count is derived from \cite{nguyen2018extragradient}.
For low-dimensional EG-FP,  the transformation of problem  (\ref{P1}) into its low-dimensional form (\ref{P5}) adds an additional computational complexity of $\mathcal{O}(K^2N_t)$ due to the computation of the matrix
 $\mathbf{G} = \mathbf{H}^H\mathbf{H}$.
Following this transformation,
the optimization framework is similar to the proposed Algorithm \ref{extragradient-FP} by replacing the channel and beamforming matrix with $\mathbf{G}$ and $\mathbf{Q}$, respectively.
Hence, the overall computational complexity for solving problem (\ref{P5}) is $\mathcal{O}(K^2N_t+\log(\epsilon_1^{-1})\sqrt{\epsilon_2^{-1}}K^3)$.

\section{Extension to Imperfect CSIT Scenarios}
\label{imperfect CSIT}
In real-world applications,
perfect CSIT is not achievable.
Imperfect CSIT is inevitable due to various sources of channel estimation errors,
such as user mobility, quantized feedback, and other practical factors.
Therefore, in this section, we extend the algorithm proposed in Section III for perfect CSIT to address the imperfect CSIT scenario,
demonstrating the robustness and applicability of our approach under realistic channel conditions.

\par 
The resource allocation problem of RSMA has been widely studied in imperfect CSIT. There are two different methods emerge to combat the influence of CSIT: worst-case resource optimization \cite{lin2021, fu2020} and long-term resource optimization\cite{yin2022, yin2021, dizdar2021}.
The former method considers bounded CSIT estimation errors,
allocating wireless resources to ensure performance across all feasible channels within corresponding uncertainty regions. This ensures relatively good performance even under the worst-case channel conditions. 
However, in scenarios where CSIT estimation errors are unbounded (e.g. when the transmitter only possesses Gaussian distribution knowledge of the CSIT estimation error),
a worst-case resource optimization becomes unavailable.
In comparison, the latter method, long-term resource optimization, smooths the impact of CSIT estimation errors by designing beamforming vectors and common rate allocation based on the averaged performance over a long-term sequence of fading channels.
\par In this section,  we focus on the long-term MMF rate optimization problem of RSMA. Different from the conventional algorithms that typically employ sample average approximation (SAA) to transform the stochastic long-term optimization problem into a deterministic one, we derive lower bounds of the ergodic rates for all streams.
Then,  we extend the algorithms proposed in Section \ref{algorithm}  to solve the corresponding lower-bound MMF rate problem in imperfect CSIT. 

\subsection{System Model with Imperfect CSIT}

We presume that the channel state, represented by $\mathbf{H} \triangleq [\mathbf{h}_1, \cdots, \mathbf{h}_K]$,
experiences variations during transmission following an ergodic stationary process with the probability density $f_\mathbf{H}(\mathbf{H})$.
The users are assumed to possess highly accurate estimations and tracking of their channel vectors, namely perfect channel state information at the receiver (CSIR).
However, the BS only knows an imperfect instantaneous channel estimate denoted as $\widehat{\mathbf H} \triangleq [\widehat{\mathbf h}_1, \cdots, \widehat{\mathbf h}_K]$.
This entire process is described by the joint distribution of $\{\mathbf{H}, \widehat{\mathbf H}\}$, which is assumed to be stationary and ergodic.
Considering a given estimate $\widehat{\mathbf{H}}$ and the estimation error matrix $\mathbf{E}$, the relationship between $\mathbf{H}$, $\widehat{\mathbf{H}}$, and $\mathbf{E}$ is:
\begin{equation}\label{channel model}
\mathbf{H} = \widehat{\mathbf{H}} + \mathbf{E},
\end{equation}
where the error $\mathbf{E}$ follows the conditional probability density $f_{\mathbf{H}|\widehat{\mathbf H}}(\mathbf{H}|\widehat{\mathbf H})$.
For each user-$k$, the marginal density of the $k$th channel is denoted as $f_{\mathbf{h}_k|\widehat{\mathbf h}_k}(\mathbf{h}_k|\widehat{\mathbf h}_k)$.
We assume the distribution's mean is represented by the estimate, i.e., $\mathbb{E}_{\mathbf{h}_k|\widehat{\mathbf h}_k}\{\mathbf{h}_k|\widehat{\mathbf h}_k\} = \widehat{\mathbf h}_k$. Additionally, we assume that  $\mathbb{E}_{\mathbf{h}_k|\widehat{\mathbf h}_k}\{\mathbf{h}_k\mathbf{h}_k^H|\widehat{\mathbf h}_k\} = \mathbf{R}_k$,
where $\mathbf{R}_k$ denotes the $k$th user’s CSIT error covariance matrix.

While the BS lacks the ability to predict instantaneous rates,
it has access to the average rates (ARs), defined as:
\begin{subequations}
\label{eq:AR}
    \begin{align}
        &\widehat{R}_{c, k}(\widehat{\mathbf H})
        \triangleq \mathbb{E}_{\mathbf{H}|\widehat{\mathbf H}}\left\{  \log\left( 1+ \frac{|{\mathbf h}_k^H\mathbf p_c|^2}
        {\sum_{j=1}^K|{\mathbf h}_k^H\mathbf p_j|^2 + \sigma_k^2} \right) \right\}, \label{AR common}\\
        &\widehat{R}_{p, k}(\widehat{\mathbf H})
        \triangleq \mathbb{E}_{\mathbf{H}|\widehat{\mathbf H}}\left\{ \log\left( 1+  \frac{|{\mathbf h}_k^H\mathbf p_k|^2}
        {\sum_{j\neq k}^K|{\mathbf h}_k^H\mathbf p_j|^2 + \sigma_k^2} \right)  \right\}. \label{AR private}
    \end{align}
\end{subequations}
For successful decoding of both common and private streams by the $k$th user over the long-term channel sequence, the BS aims to transmit these messages at ergodic rates (ERs), which can be obtained by taking the average of ARs over the entire domain of $\widehat{\mathbf H}$ as \cite{Joudeh2016sumrate}:
\begin{subequations}
    \begin{align}
        &\overline{R}_{c, k} = \mathbb{E}_{\widehat{\mathbf H}} \left\{ \widehat{R}_{c, k}(\widehat{\mathbf H}) \right\}, \\
        &\overline{R}_{p, k} = \mathbb{E}_{\widehat{\mathbf H}} \left\{ \widehat{R}_{p, k}(\widehat{\mathbf H}) \right\}.
    \end{align}
\end{subequations}
For successful SIC, the code rate of the common stream $s_c$ is set as the minimum among the users' ERs, resulting in the ER of the common message as:
\begin{equation}
    \begin{aligned}
        \overline{R}_c \triangleq \min_{k\in \mathcal{K}}\left\{ \overline{R}_{c, k} \right\}.
    \end{aligned}
\end{equation}

\subsection{Problem Formulation and Optimization Framework with Imperfect CSIT}
The primary challenge in handling ER-based resource optimization problems arises from the complexity introduced by the expectation operation, preventing the derivation of closed-form expressions for the ERs. 
A conventional approach to dealing with ERs is SAA,
which is introduced by Joudeh et al. \cite{Joudeh2016sumrate} for RSMA.
Such approach requires a substantial number of channel samples from the channel distribution and formulates problems based on these channel samples. However, this method is highly complex due to the large amount of channel samples and the computations involved.
In this subsection, to overcome this limitation,
we follow \cite{korean} and derive lower bounds for the ERs.
These lower bounds offer deterministic closed-form expressions,
which are more tractable within our proposed optimization framework.

\par
The AR $\widehat{R}_{c, k}$ defined in (\ref{AR common}) is constrained by the following lower bounds:
\begin{equation}\label{lower bound derivation}
    \begin{aligned}
        \widehat{R}_{c, k}
        &= \mathbb{E}_{\mathbf{H}|\widehat{\mathbf H}}\left\{  \log\left( 1+ \frac{|{\mathbf h}_k^H\mathbf p_c|^2}
        {\sum_{j=1}^K|{\mathbf h}_k^H\mathbf p_j|^2 + \sigma_k^2} \right) \right\}\\
        &= \mathbb{E}_{\mathbf{H}|\widehat{\mathbf H}}\left\{  \log\left( 1+ \frac{|{(\widehat{\mathbf h}}_k + {\mathbf e}_k)^H\mathbf p_c|^2}
        {\sum_{j=1}^K|(\widehat{\mathbf h}_k + {\mathbf e}_k)^H\mathbf p_j|^2 + \sigma_k^2} \right) \right\} \\
        &\overset{(a)}{\geq} \mathbb{E}_{\mathbf{H}|\widehat{\mathbf H}}\left\{  \log\left( 1+ \frac{|{\widehat{\mathbf h}}_k^H\mathbf p_c|^2}
        {\begin{array}{ll}
            \sum_{j=1}^K|\widehat{\mathbf h}_k^H\mathbf p_j|^2 + |{\mathbf e}_k^H\mathbf{p}_c|^2 \\
            + \sum_{j=1}^K|\mathbf{e}_k^H\mathbf{p}_j|^2 + \sigma_k^2
        \end{array}
        } \right) \right\} \\
        & \overset{(b)}{\geq} \log\left( 1+ \frac{|{\widehat{\mathbf h}}_k^H\mathbf p_c|^2}
        {\begin{array}{l}
            \sum_{j=1}^K|\widehat{\mathbf h}_k^H\mathbf p_j|^2 + \mathbf{p}_c^H\mathbb{E}_{\mathbf{H}|\widehat{\mathbf H}}\{\mathbf{e}_k\mathbf{e}_k^H\}\mathbf{p}_c \\
            + \sum_{j=1}^K\mathbf{p}_j^H\mathbb{E}_{\mathbf{H}|\widehat{\mathbf H}}\{\mathbf{e}_k\mathbf{e}_k^H\}\mathbf{p}_j + \sigma_k^2
        \end{array}
        } \right) \\
        & = \log\left( 1+ \frac{|{\widehat{\mathbf h}}_k^H\mathbf p_c|^2}
        {\begin{array}{l}
            \sum_{j=1}^K|\widehat{\mathbf h}_k^H\mathbf p_j|^2 + \mathbf{p}_c^H\mathbf{R}_k\mathbf{p}_c \\
            + \sum_{j=1}^K\mathbf{p}_j^H\mathbf{R}_k\mathbf{p}_j + \sigma_k^2
        \end{array}
        } \right) \triangleq \widehat{R}_{c, k}^{lb},
    \end{aligned}
\end{equation}
where inequality (a) is obtained from treating $\mathbf{e}_k^H\mathbf{p}_cs_c + \sum_{j=1}^K \mathbf{e}_k^H \mathbf{p}_j s_j$ as independent Gaussian noise,
and inequality (b) is derived by applying Jensen's inequality \cite{convex}.
Similarly, the AR $\widehat{R}_{p, k}$ defined in (\ref{AR private}) is constrained by:
\begin{equation}
    \begin{aligned}
        \widehat{R}_{p, k}
        &= \mathbb{E}_{\mathbf{H}|\widehat{\mathbf H}}\left\{  \log\left( 1+ \frac{|{\mathbf h}_k^H\mathbf p_k|^2}
        {\sum_{j\neq k}^K|{\mathbf h}_k^H\mathbf p_j|^2 + \sigma_k^2} \right) \right\} \\
        & = \log\left( 1+ \frac{|{\widehat{\mathbf h}}_k^H\mathbf p_k|^2}
        {
            \sum_{j\neq k}^K|\widehat{\mathbf h}_k^H\mathbf p_j|^2 
            + \sum_{j=1}^K\mathbf{p}_j^H\mathbf{R}_k\mathbf{p}_j + \sigma_k^2
        } \right) \\
        &\triangleq \widehat{R}_{p, k}^{lb}.
    \end{aligned}
\end{equation}

Using inequality (\ref{lower bound derivation}), the ER of the common stream is therefore bounded by:
\begin{equation}
    \begin{aligned}
        \overline{R}_{c}
        &= \min\limits_{k \in \mathcal{K}} \left\{ \mathbb{E}_{\widehat{\mathbf H}} \left\{ \widehat{R}_{c, k}(\widehat{\mathbf H}) \right\} \right\}\geq \min\limits_{k \in \mathcal{K}} \left\{ \mathbb{E}_{\widehat{\mathbf{H}}}\left\{\widehat{R}_{c, k}^{lb}(\widehat{\mathbf H})\right\} \right\} \\
        &\overset{(c)}{\geq} \mathbb{E}_{\widehat{\mathbf{H}}} \left\{ \min\limits_{k \in \mathcal{K}} \left\{\widehat{R}_{c, k}^{lb}(\widehat{\mathbf H})\right\} \right\},
    \end{aligned}
\end{equation}
where inequality (c) is derived from Jensen's inequality,
leveraging the concavity of the function $\min\{\cdot\}$.
Similarly, the ER of each private stream is bounded by:
\begin{equation}
     \overline{R}_{p, k}
     = \mathbb{E}_{\widehat{\mathbf{H}}}\left\{\widehat{R}_{p, k}(\widehat{\mathbf H})\right\}
     \geq \mathbb{E}_{\widehat{\mathbf{H}}}\left\{\widehat{R}_{p, k}^{lb}(\widehat{\mathbf H})\right\}.
\end{equation}
Therefore, the stochastic ERs are reduced to deterministic closed-form lower bound expressions. By replacing $R_c$ and $R_{p, k}$ in problem (\ref{P1}) with $\min_{k \in \mathcal{K}} {\widehat{R}_{c, k}^{lb}(\widehat{\mathbf H})}$ and $\widehat{R}_{p, k}^{lb}(\widehat{\mathbf H})$, we formulate the  MMF problem in imperfect CSIT as follows:
\begin{subequations}\label{P6}
    \begin{align}
        &\max_{\mathbf P, \mathbf c} \min\limits_{k\in\mathcal{K}}\left\{\widehat{R}_{p, k}^{lb}(\widehat{\mathbf H}) + c_k\right\} \label{P6 0} \\
        \mathrm{s.t.}\ & \mathbf{1}^T\mathbf{c} \leq \min\limits_{k \in \mathcal{K}} \left\{\widehat{R}_{c, k}^{lb}(\widehat{\mathbf H})\right\}, \label{P6 1}\\
        &\mathbf{c} \succeq \mathbf{0}, \label{P6 2}\\
        &\mathrm{tr}\left(\mathbf{P}^H\mathbf{P}\right) \leq P_t.
    \end{align}
\end{subequations}

The structure of problem (\ref{P6}) closely resembles that of problem (\ref{P1}),
allowing us to employ both EG-FP and low-dimensional EG-FP algorithms to solve the problem.
A detailed explanation of the corresponding optimization framework is omitted here to eliminate redundancy.

\section{Numerical Results}
\label{numerical results}
In this section,
we evaluate the performance of the proposed EG-FP and low-dimensional EG-FP algorithms against two other baseline schemes: namely, SCA \cite{yalcin2021, joudeh2017rate} and GPI \cite{korean}, as described in Section \ref{system model}.
The WMMSE algorithm is not considered here since it is known as a special instance of SCA, yielding nearly identical performance  \cite{Joudeh2016sumrate, SEandEE2019}.
It is worth noting that the computational complexity of WMMSE for RSMA, given by $\mathcal{O}(\log (\epsilon^{-1}) (KN_t)^{3.5})$ \cite{SEandEE2019},
is of the same order of magnitude as that of the SCA-based algorithm, resulting in comparable CPU times, as observed in our previous work \cite{luo2023practical}.
The algorithms considered in the numerical results are outlined below:
\begin{enumerate}
    \item \textbf{EG-FP}: This is the algorithm we proposed in Algorithm \ref{extragradient-FP}. It has a worst-case computational complexity of $\mathcal{O}\left(\log(\epsilon_1^{-1})\sqrt{\epsilon_2^{-1}}K^2N_t\right)$.

    \item \textbf{Low-dimensional EG-FP}:
    This is the algorithm we proposed in Section \ref{low dimensional form}. Its worst-case computational complexity is $\mathcal{O}\left(\log(\epsilon_1^{-1})\sqrt{\epsilon_2^{-1}}K^3+K^2N_t\right)$.

    \item \textbf{SCA}:
    This algorithm, as per \cite{mao2018energy}, has a worst-case computational complexity of $\mathcal{O}(\log(\epsilon^{-1})(KN_t)^{3.5})$.
    \item \textbf{GPI}:
    This algorithm, proposed in \cite{korean},
    demonstrates a worst-case computational complexity of $\mathcal{O}((\#\gamma)\log(\epsilon^{-1})KN_t^{3})$,
    where $\#\gamma$ denotes the exhaustive search count for the Lagrangian multiplier $\gamma$ (as specified in \cite{korean}).
\end{enumerate}
Only SCA requires to use an optimization toolbox in each iteration to solve the corresponding convex subproblem.
All SCA simulations in this section utilize the widely adopted optimization toolbox CVX \cite{grant2008cvx}.
The other three algorithms are all optimization toolbox-free algorithms.
As EG-FP has higher computational complexity than the low-dimensional EG-FP when $N_t>K$,
but both achieve locally optimal solutions,
in the following analysis,
EG-FP is employed when $N_t\leq K$,
while the low-dimensional EG-FP is adopted when $N_t>K$.

In all numerical results,
the noise variance is uniformly set to $\sigma_k^2 = 1$ for all $k\in\mathcal{K}$. Therefore, the transmit power $P_t$ is numerically equivalent to the transmit SNR.
The convergence tolerance for all iteration loops is fixed to  $10^{-3}$ for all schemes.
Unless stated otherwise,
all results are averaged over 100 random channel realizations.

\begin{figure}[htbp]
    \centering
    \subfloat[\label{convergence perfect}Perfect CSIT]{\includegraphics[width=0.24\textwidth]{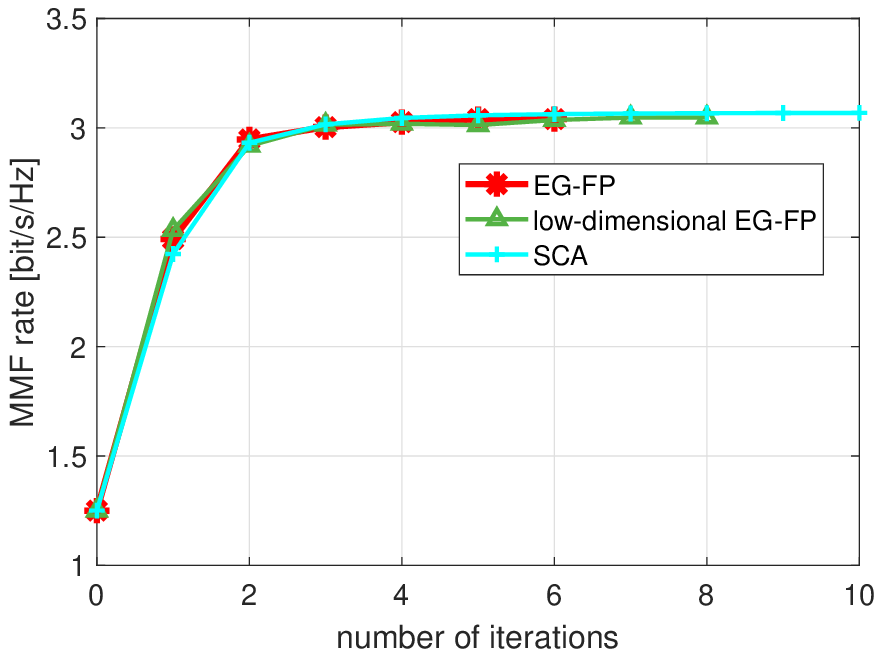}}
    \hfill
    \subfloat[\label{convergence imperfect}Imperfect CSIT]{\includegraphics[width=0.24\textwidth]{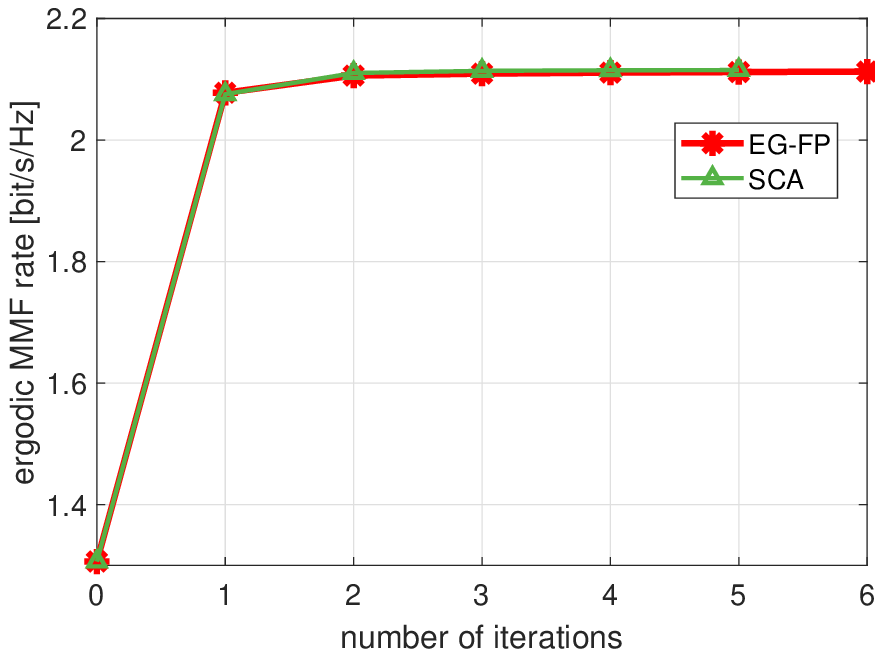}}
    \caption{Convergence of different algorithms in perfect and imperfect CSIT, where $K = 4$ and $N_t = 16$.}
    \label{convergence}
\end{figure}    

We start by examining the convergence behavior of the proposed algorithms.
Fig. \ref{convergence}(a) shows the MMF rate against the number of iterations in a specific channel realization with $K = 4$ and $N_t = 16$ in the perfect CSIT scenario.
It is worth noting that the outer layer of GPI involves an exhaustive search for the optimal Lagrangian multiplier,
which is not an iterative process.
GPI is therefore not included in Fig. \ref{convergence}.
We can observe from the figure that the proposed EG-FP and low-dimensional EG-FP algorithms achieve MMF rates almost identical to that of the SCA algorithm.
This aligns with our theoretical analysis.
Despite converging within almost the same number of steps as the SCA algorithm,
EG-FP and low-dimensional EG-FP attain much lower computational complexity per step,
as we will illustrate shortly.
This is the primary reason EG-FP executes faster than SCA.
Fig. \ref{convergence}(b) shows the convergence behavior of the algorithms in a specific channel realization with $K = 4$ and $N_t = 16$ in the imperfect CSIT scenario,
which shows a similar convergence behavior to the perfect CSIT scenario.

\begin{figure}[htbp]
    \centering
    \subfloat[\label{cumulative perfect}Perfect CSIT]{\includegraphics[width=0.24\textwidth]{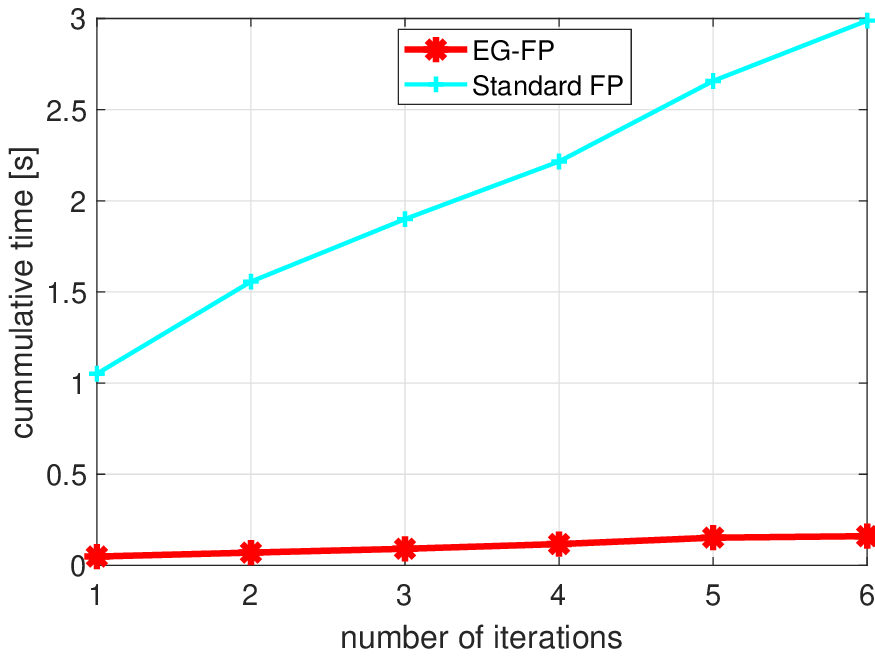}}
    \hfill
    \subfloat[\label{cumulative imperfect}Imperfect CSIT]{\includegraphics[width=0.24\textwidth]{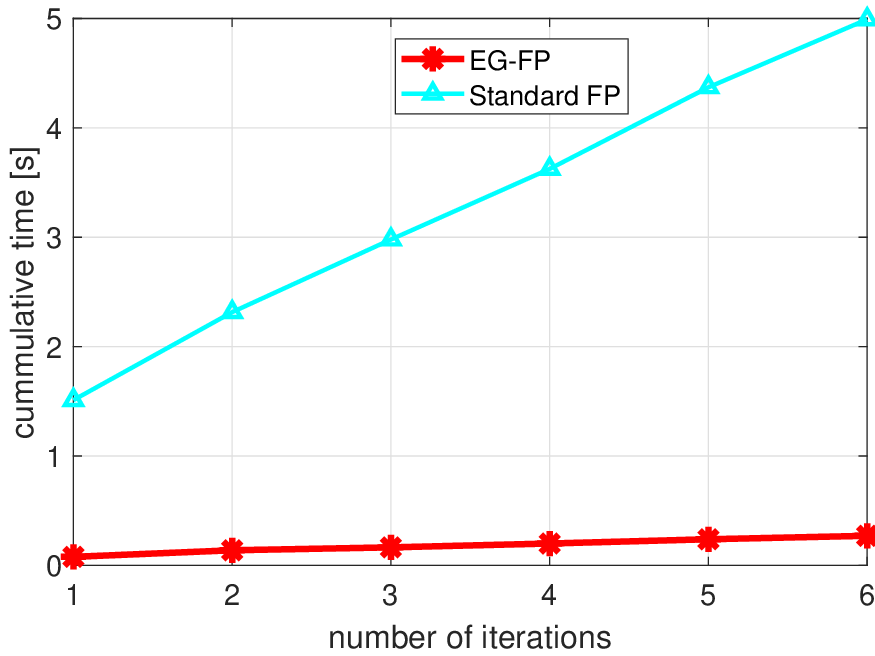}}
    \caption{Cumulative time of different algorithms in perfect and imperfect CSIT, where $K = 4$ and $N_t = 16$.}
    \label{cumulative}
\end{figure}

Fig. \ref{cumulative} illustrates a detailed comparison of the cumulative computational time required by the proposed EG-FP algorithm and the standard FP algorithm, which uses CVX to solve problem (\ref{P3}). Since both algorithms use the same FP framework to transform the original RSMA MMF problem into a series of convex beamforming subproblems but differ in subproblem solvers, this comparison shows the impact of our proposed EG.
The simulation settings are the same as Fig. \ref{convergence}.
Across both scenarios, EG-FP demonstrates a clear superiority in terms of cumulative computational time.
Its cumulative time remains much lower than that of the standard FP,
often by an order of magnitude,
making it a highly efficient choice for MMF problem of RSMA.
Moreover, as the number of iterations increases,
the cumulative time of the standard FP rises faster compared to EG-FP.

\subsection{Perfect CSIT}

In all scenarios with perfect CSIT,
the channel between the BS and user $k$ follows an i.i.d. complex Gaussian distribution denoted as $\mathbf{h}_k \sim \mathcal{CN}(\mathbf{0}, \mathbf{I})$.
The transmit signal-to-noise ratio (SNR) is fixed to $10$ dB.

In Fig. \ref{K_equals_Nt}(a),
the average MMF rate versus the number of transmit antennas or users is illustrated, where $N_t = K$. 
{Due to the increasing computational
cost of SCA and GPI with the number of transmit antennas,
we exclusively present their results for scenarios involving 2 to
16 transmit antennas and users.}
It demonstrates that the proposed EG-FP achieves almost the same performance as the SCA algorithm in both small-scale and large-scale networks.  Algorithm \ref{extragradient-FP} effectively solves the original MMF problem and achieves sub-optimal solutions.
In contrast, the MMF rate achieved by the GPI algorithm is notably worse than that of the EG-FP approach. 
The issue arises from the utilization of a fixed common rate allocation in GPI at its initial stage. Even if GPI re-allocates the common rate at the second stage, the beamforming vectors have already been fixed, restricting the potential adjustments to the common rate allocation.

\par
The average CPU time versus the number of transmit antennas or users is illustrated in Fig. \ref{K_equals_Nt}(b) under the same simulation setting as Fig. \ref{K_equals_Nt}(a).
It shows that EG-FP significantly reduces the CPU time compared to the SCA algorithm.
Specifically, the average computational time of EG-FP is nearly three orders of magnitude lower than that of the SCA algorithm.
Although EG-FP takes longer runtime than GPI when $K = 2$ and $K = 4$,
the resulting MMF rate of EG-FP is much higher than GPI.
Although each EG-FP inner iteration is computationally efficient, hundreds of them occur per outer iteration.
For GPI's exhaustive search, we set the step size of the Lagrangian multiplier to 2 over the range [0, 50], resulting in lower average CPU time for fewer than 4 users.
Moreover, the CPU time of the GPI algorithm increases rapidly as the number of users increases.

\begin{figure}[htbp]
    \centering
    \subfloat[\label{K_equals_Nt_rate}Average MMF rate]{\includegraphics[width=0.24\textwidth]{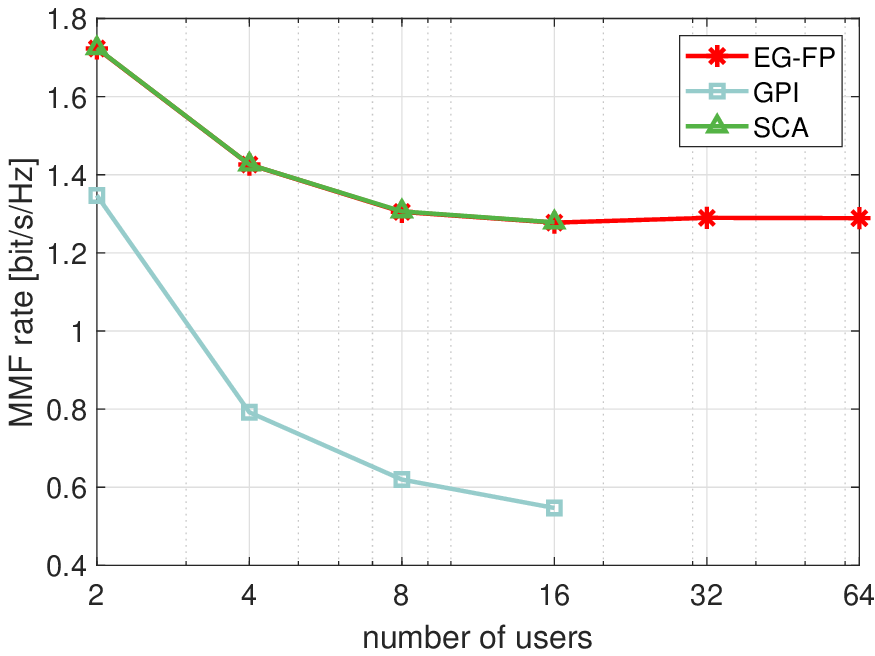}}
    \hfill
    \subfloat[\label{K_equals_Nt_time}Average CPU time]{\includegraphics[width=0.24\textwidth]{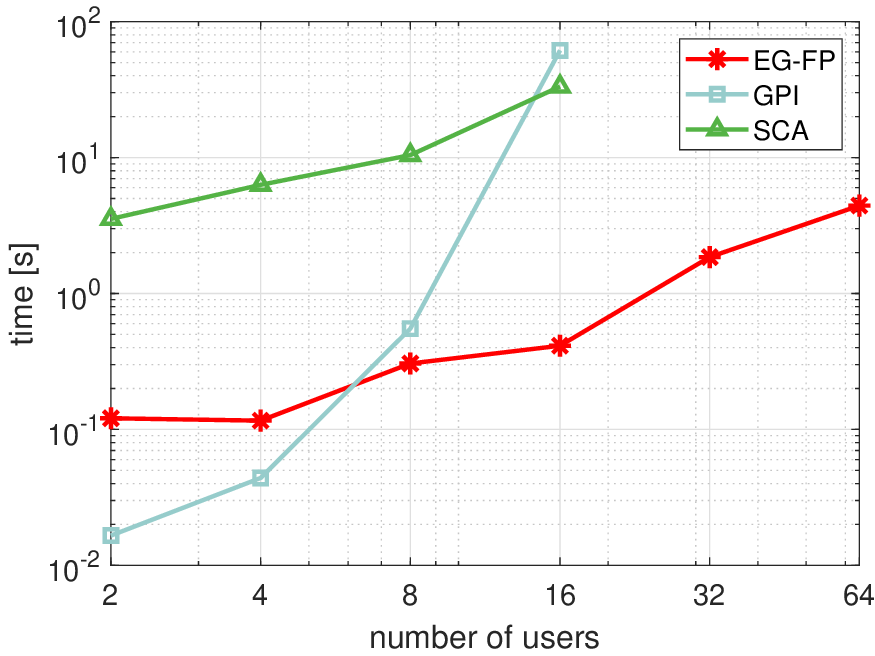}}
    \caption{Average MMF rate and CPU time versus the number of transmit antennas or users, where $N_t = K$.}
    \label{K_equals_Nt}
\end{figure}

\begin{figure}[htbp]
    \centering
    \subfloat[\label{Nt_varies_rate}Average MMF rate]{\includegraphics[width=0.24\textwidth]{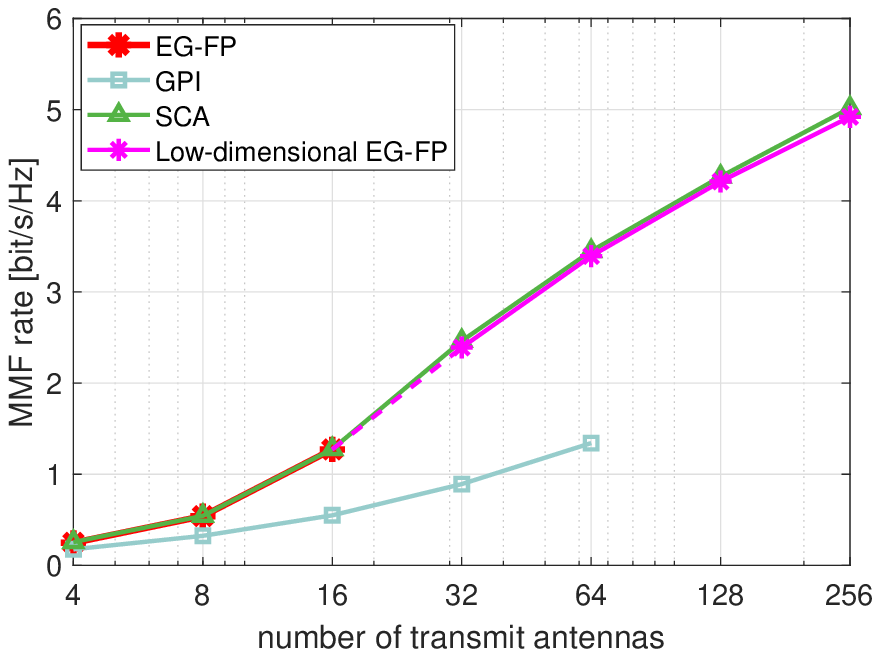}}
    \hfill
    \subfloat[\label{Nt_varies_time}Average CPU time]{\includegraphics[width=0.24\textwidth]{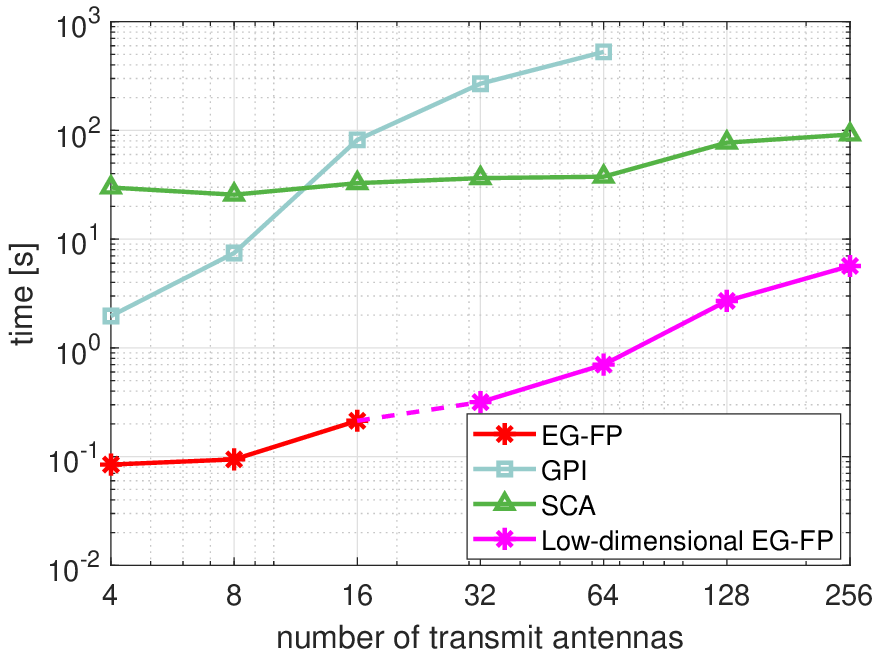}}
    \caption{Average MMF rate and CPU time versus the number of transmit antennas, where $K =16$.}
    \label{Nt_varies}
\end{figure}

\par 
In Fig. \ref{Nt_varies}(a), we illustrate the average MMF rates of different algorithms versus the number of transmit antennas.
The user number is fixed to 16. 
In the scenario of $N_t \leq K$, we use EG-FP while in the scenario of $N_t > K$, we use the low-dimensional EG-FP instead.
We only illustrate the results of GPI for scenarios with 4 to 64 transmit antennas due to its high computational complexity.
The results also show that the proposed algorithms achieve nearly the same MMF rate as the SCA algorithm for all different numbers of transmit antennas.
In contrast, the MMF rate of GPI falls below that of EG-FP, indicating that GPI cannot guarantee locally optimal solutions.

\par
Fig. \ref{Nt_varies}(b) illustrates the average CPU time for various algorithms versus the number of transmit antennas,
with the same simulation setting as in Fig. \ref{Nt_varies}(a).
Although the average CPU time of the SCA increases slowly with the number of transmit antennas, the overall time scale remains quite large.
The proposed two algorithms only take less than 7\% average CPU time of SCA.
The proposed low-dimensional EG-FP efficiently manages massive MIMO scenarios,
requiring only 5.6638 seconds to solve the MMF resource optimization problem of RSMA with little performance loss when the number of transmit antennas is 256.

\begin{figure}[htbp]
    \centering
    \subfloat[\label{low_dim_vs_normal_rate}Average MMF rate]{\includegraphics[width=0.49\linewidth]{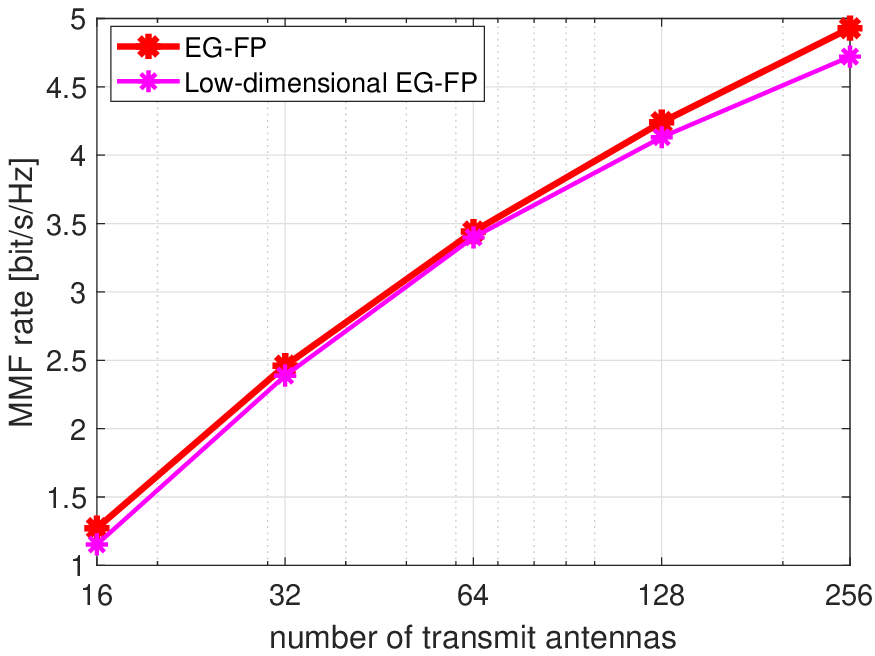}}
    \hfill
    \subfloat[\label{low_dim_vs_normal_time}Average CPU time]{\includegraphics[width=0.49\linewidth]{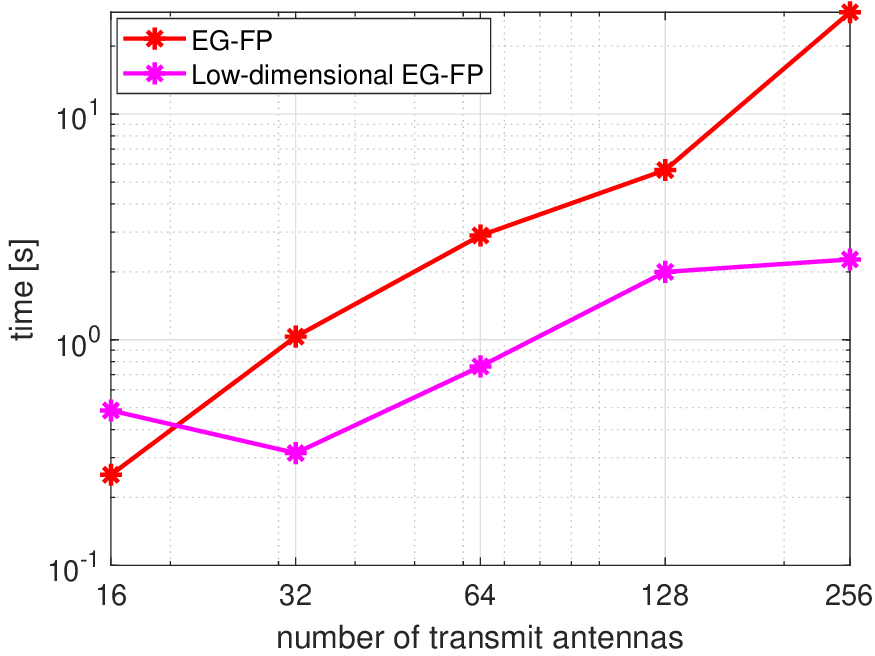}}
    \caption{Average MMF rate and CPU time versus the number of transmit antennas, where $K =16$.}
    \label{low_dim_vs_normal}
\end{figure}

Fig. \ref{low_dim_vs_normal} compares the performance of the proposed EG-FP and low-dimensional EG-FP algorithms when the number of transmit antennas $N_t$ increases.
Since the low-dimensional EG-FP algorithm is introduced in Section III-E for scenarios where $N_t>K$,
the results are illustrated exclusively for these scenarios.
The results show that the low-dimensional EG-FP algorithm achieves almost the same MMF rate as the EG-FP algorithm for all different numbers of transmit antennas.
The average CPU time of the low-dimensional EG-FP algorithm is lower than that of the EG-FP algorithm except for the case of $K = N_t = 16$.
This is because the low-dimensional EG-FP algorithm has to do the matrix multiplication $\mathbf{H}^H\mathbf{H}$.
It is more efficient than EG-FP only when $N_t>K$.

\subsection{Imperfect CSIT}

\begin{figure}[htbp]
    \centering
    \subfloat[\label{Nt_varies_rate_imperfect}Average ergodic MMF rate]{\includegraphics[width=0.24\textwidth]{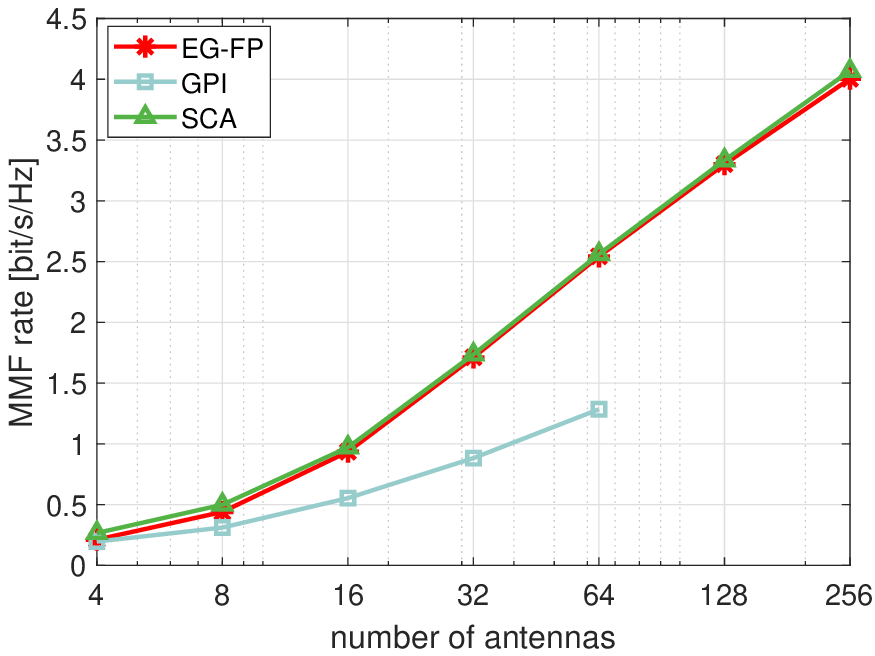}}
    \hfill
    \subfloat[\label{Nt_varies_time_imperfect}Average CPU time]{\includegraphics[width=0.24\textwidth]{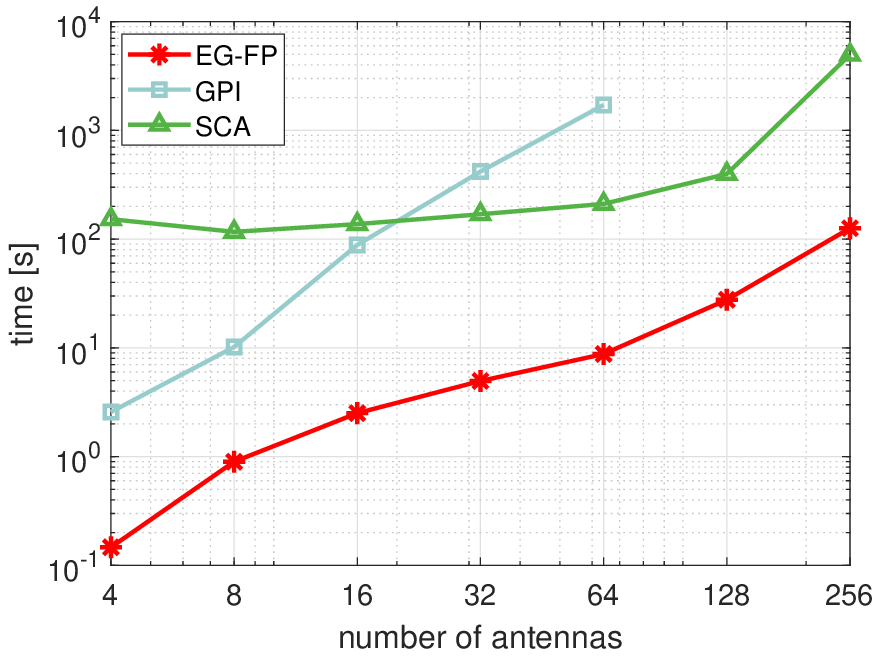}}
    \caption{Average ergodic MMF rate and CPU time versus the number of transmit antennas, where $K = 16$.}
    \label{Nt_varies_imperfect}
\end{figure}

\begin{figure}[htbp]
    \centering
    \subfloat[\label{imperfect_rate}Average ergodic MMF rate]{\includegraphics[width=0.24\textwidth]{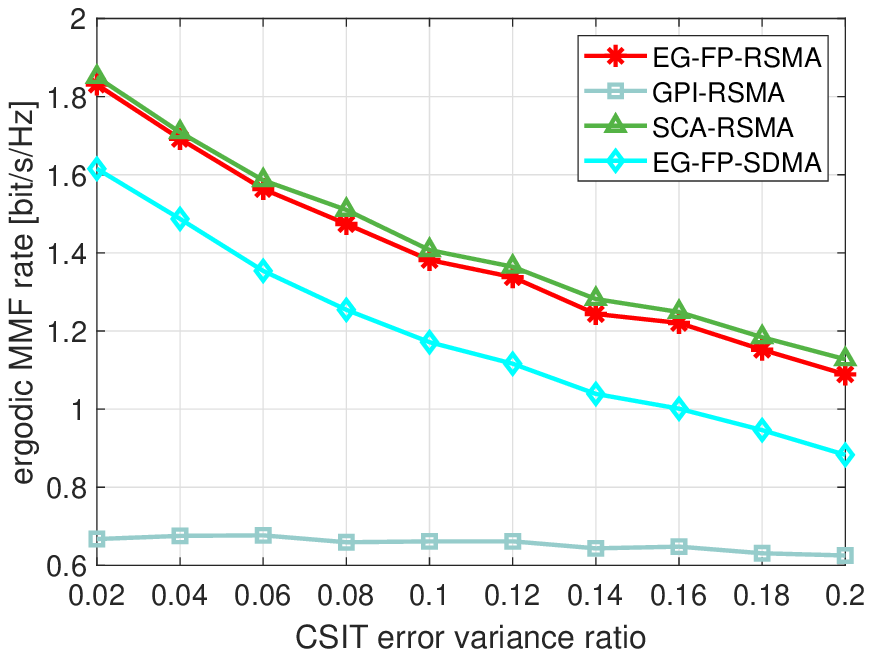}}
    \hfill
    \subfloat[\label{imperfect_time}Average CPU time]{\includegraphics[width=0.24\textwidth]{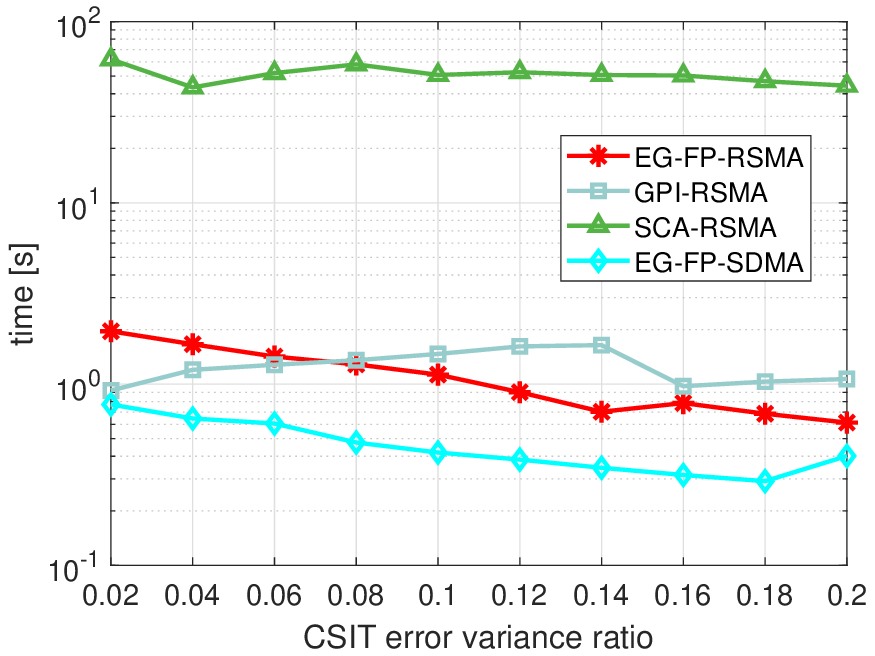}}
    \caption{Average ergodic MMF rate and CPU time versus CSIT error variance ratio,  where $N_t = K = 8$.}
    \label{imperfect}
\end{figure}

In this subsection,
we investigate the impact of imperfect CSIT.
The transmit signal-to-noise ratio (SNR) is fixed to $20$ dB.
The user channels are modeled using the same approach as outlined in \cite{clerckx2021noma}.
Specifically, based on the imperfect CSIT model specified in (\ref{channel model}),
the channel estimation $\widehat{\mathbf h}_k$ is generated by following a complex normal distribution $\mathcal{CN}(\mathbf{0}, (1 - \kappa)\sigma_k^2\mathbf{I})$,
while the estimation error $\mathbf{e}_k|\widehat{\mathbf h}_k$ follows $\mathcal{CN}(\mathbf{0}, \kappa\sigma_k^2\mathbf{I})$.
Here, $\sigma_k^2$ represents the channel variance of user-$k$.
It is randomly generated from the interval $[0.1, 1]$ to capture diverse path-loss conditions.
For convenience, we refer to $\kappa$ as the CSIT error variance ratio.

We first investigate the impact of the number of transmit antennas on the performance of the proposed algorithm.
Fig. \ref{Nt_varies_imperfect}(a) illustrates the average MMF rate versus the number of transmit antennas,
where the number of users is fixed to 16.
The missing results for GPI when the number of transmit antennas exceeds 64 are due to its exceptionally high computational complexity at larger antenna configurations.
The results show that the proposed EG-FP algorithm achieves almost the same MMF rate as the SCA algorithm for all different numbers of transmit antennas.
The MMF rate of GPI is notably worse than that of the proposed algorithm.
Fig. \ref{Nt_varies_imperfect}(b) illustrates the average CPU time for various algorithms versus the number of transmit antennas,
with the same simulation setting as in Fig. \ref{Nt_varies_imperfect}(a).
The proposed EG-FP only takes less than 9\% average CPU time of SCA and GPI.
Compared to the perfect CSIT scenario, all algorithms take more CPU time in the imperfect CSIT scenario.
This is because the denominator of the rate lower bound expression becomes more complex due to the presence of the CSIT error.

Besides GPI and SCA, we also compare the performance of RSMA with SDMA.
Since SDMA does not have any common stream,
by disabling the common rate allocation and the precoder of the common stream (i.e., setting $\mathbf{c}=\bm 0$ and $\mathbf{p}_c=\bm 0$),
the MMF problem of RSMA reduces to the MMF problem of SDMA as:
\begin{subequations} \label{sdma}
    \begin{align}
        &\max_{\mathbf P_p} \min\limits_{k\in\mathcal{K}}\left\{\widehat{R}_{p, k}^{lb}(\widehat{\mathbf H})\right\} \\
        \mathrm{s.t.}\ & \mathrm{tr}\left(\mathbf{P}_p^H\mathbf{P}_p\right) \leq P_t, 
    \end{align}
\end{subequations}
where $\mathbf{P}_p=[\mathbf{p}_1,\cdots,\mathbf{p}_K] \in \mathbb{C}^{N_t\times K}$, and $\widehat{R}_{p, k}^{lb}(\widehat{\mathbf H})$ is defined in  (\ref{lower bound derivation}).
The proposed EG-FP algorithm can be directly applied to address (\ref{sdma}) since it is a special instance of the RSMA problem.

Fig. \ref{imperfect}(a) and Fig. \ref{imperfect}(b) respectively illustrate the ER and the average CPU time versus the CSIT error variance ratio $\kappa$,
where $N_t = K = 8$.
Similar to the scenario with perfect CSIT,
the ER of EG-FP closely aligns with that of SCA, and it outperforms GPI significantly.
It only takes less than 5\%  average CPU time of SCA.
While the CPU time of GPI is also low in this small-scale network,
its ergodic MMF rate performance is much worse than EG-FP and SCA.
This coincides with the result of perfect CSIT.
Under the same CSIT error variance ratio,
it is evident that the performance of RSMA outperforms SDMA.
This demonstrates the superior interference management capability of RSMA compared to SDMA in the presence of imperfect CSIT.

\section{Conclusion}
\label{conclusion}
This paper introduces two novel and efficient resource optimization algorithms for addressing the MMF problem of RSMA, namely, the EG-FP algorithm and its low-dimensional variant. 
The proposed EG-FP algorithm first converts the original problem into a sequence of convex subproblems based on FP, which are solvable within an AO framework.
By exploring the Lagrangian dual problem of each convex subproblem, we identify the variational inequality property for the dual problem.
The primary novelty of the proposed algorithm lies in effectively utilizing the variational inequality property to address each subproblem through the introduced extragradient algorithm.
Additionally, we propose a low-dimensional beamforming design, particularly beneficial in scenarios with a higher number of transmit antennas than users.
To expand the research scope, we extend the proposed algorithms to solve MMF rate problems of RSMA with imperfect CSIT.
Through extensive numerical experiments,
we demonstrate that our proposed algorithms achieve MMF rates nearly equivalent to the conventional SCA algorithm and significantly outperform the GPI baseline scheme.
Notably, our proposed algorithm exhibits a remarkable reduction in CPU time, amounting to less than 10\% of the SCA algorithm's runtime.
This efficiency arises from the closed-form updates in each inner iteration.
In contrast, SCA relies on the computationally intensive CVX toolbox to solve optimization problems,
and GPI employs exhaustive search methods, which are inherently time-consuming and less efficient.
This advantage makes our algorithm as a promising approach for the practical and efficient design of resource optimization algorithms for RSMA, showing its potential applications in 6G networks.

\bibliographystyle{IEEEtran}{}
\bibliography{reference}

\appendices
\section{Proof of Theorem \ref{outer convergence}}\label{outer convergence proof}
\begin{proof}
    The outer loop of Algorithm \ref{extragradient-FP}  alternatively optimizes three variable blocks  $\{\mathbf{P}, \mathbf{c}\}$, $\{\boldsymbol{\vartheta}_c, \boldsymbol{\vartheta}_p\}$, $\{\boldsymbol{\varphi}_c, \boldsymbol{\varphi}_p\}$ for problem (\ref{P2}).

    Denote $F(\mathbf{P}, \mathbf{c}, \boldsymbol{\vartheta}_c, \boldsymbol{\vartheta}_p, \boldsymbol{\varphi}_c, \boldsymbol{\varphi}_p)$ as the objective function of problem (\ref{P2}).
    Then we have
    \begin{subequations}
    \begin{align}
        &F(\mathbf{P}^{(n)}, \mathbf{c}^{(n)}, \boldsymbol{\vartheta}_c^{(n)}, \boldsymbol{\vartheta}_p^{(n)}, \boldsymbol{\varphi}_c^{(n)}, \boldsymbol{\varphi}_p^{(n)}) \\
        \overset{(d)}{\leq} &F(\mathbf{P}^{(n+1)}, \mathbf{c}^{(n+1)}, \boldsymbol{\vartheta}_c^{(n)}, \boldsymbol{\vartheta}_p^{(n)}, \boldsymbol{\varphi}_c^{(n)}, \boldsymbol{\varphi}_p^{(n)}) \\
        \overset{(e)}{\leq} &F(\mathbf{P}^{(n+1)}, \mathbf{c}^{(n+1)}, \boldsymbol{\vartheta}_c^{(n+1)}, \boldsymbol{\vartheta}_p^{(n+1)}, \boldsymbol{\varphi}_c^{(n+1)}, \boldsymbol{\varphi}_p^{(n+1)}).
    \end{align}
    \end{subequations}
    Inequality (d) holds because $(\mathbf{P}^{(n+1)}, \mathbf{c}^{(n+1)})$ is the optimal solution with the auxiliary variables $(\boldsymbol{\vartheta}_c^{(n)}, \boldsymbol{\vartheta}_p^{(n)}, \boldsymbol{\varphi}_c^{(n)}, \boldsymbol{\varphi}_p^{(n)})$ obtained from the previous iteration $n$,
    while inequality (e) holds because
    \begin{subequations}
    \begin{align}
        &g_{c,k}(\mathbf{P}^{(n+1)},\vartheta_{c,k},\varphi_{c,k})
        \leq g_{c,k}(\mathbf{P}^{(n+1)},\vartheta_{c,k}^{(n+1)},\varphi_{c,k}^{(n+1)}), \\
        &g_{p,k}(\mathbf{P}^{(n+1)},\vartheta_{p,k},\varphi_{p,k})
        \leq g_{p,k}(\mathbf{P}^{(n+1)},\vartheta_{p,k}^{(n+1)},\varphi_{p,k}^{(n+1)}).
    \end{align}
    \end{subequations}

    The above analysis shows that after one AO step  in the outer loop,
    the objective function is ensured to be non-decreasing.
    Considering that  problem (\ref{P2}) should adhere to the power constraint $\mathrm{tr}\left(\mathbf{P}^H\mathbf{P}\right) \leq P_t$,
    the objective function is bounded above.
    This implies that the outer loop of Algorithm \ref{extragradient-FP} is guaranteed to converge.
\end{proof}

\ifCLASSOPTIONcaptionsoff
\newpage
\fi

\end{document}